\documentclass[10pt]{amsart}
\usepackage{amssymb,amsmath,amscd,amsfonts,amsthm,mathrsfs}
\usepackage{verbatim} 
\input xy
\xyoption{all}
\usepackage[all]{xy}  
\usepackage{color}

\usepackage[latin1]{inputenc}

\newcommand{\Z}{\mathbb{Z}}

\newcommand{\R}{\mathbb{R}}
\newcommand{\C}{\mathbb{C}}

\newcommand{\Bc}{\mathcal{B}}
\newcommand{\Cc}{\mathcal{C}}

\newcommand{\Hc}{\mathcal{H}}

\newcommand{\demi}{\mathbb{H}}

\newcommand{\ess}{{\rm{ess}}}
\newcommand{\dhlim}{d_\demi\!\! - \!\!\!\lim}

\newcommand{\NZ}{\mathbb{G}}
\newcommand{\td}{\tilde{d}}
\newcommand{\pprime}{{\prime \prime}}

\def\build#1_#2^#3{\mathrel{\mathop{\kern 0pt#1}\limits_{#2}^{#3}}}

\newtheorem{theorem}{Theorem}[section]
\newtheorem{proposition}[theorem]{Proposition}
\newtheorem{lemma}[theorem]{Lemma}
\newtheorem{remark}[theorem]{Remark}
\newtheorem{corollary}[theorem]{Corollary}

\numberwithin{equation}{section}

\def \e{\varepsilon}

\def \bone{\mathbf{1}}

\def \rmi{{\rm i}}

\title{On the a.c.\ spectrum of the 1D discrete Dirac operator}

\begin{document}
\author{Sylvain Gol\'enia}
\address{Institut de Mathematiques de Bordeaux
Universite Bordeaux 1
351, cours de la Lib\'eration
F-33405 Talence cedex}
\email{sylvain.golenia@math.u-bordeaux1.fr}
\author{Tristan Haugomat}
\address{Universit\'e de Rennes 1,
263 avenue du Général Leclerc
CS 74205 - 35042 RENNES CEDEX \\
\'Ecole normale sup\'erieure de Cachan
Antenne de Bretagne,
Campus de Ker Lann 
Avenue Robert Schuman
35170 Bruz - France}
\email{tristan.haugomat@etudiant.univ-rennes1.fr}
\begin{abstract}
In this paper, under some integrability condition, we prove 
that an electrical perturbation of the discrete Dirac operator has 
purely absolutely continuous spectrum for the one dimensional
case. We reduce the problem to a non-self-adjoint Laplacian-like
operator by using a spin up/down decomposition and rely on a transfer
matrices technique.
\end{abstract}

\maketitle

\tableofcontents

\section{Introduction}
We study properties of relativistic (massive or not) charged particles with
spin-$1/2$. We follow the Dirac formalism, see \cite{Di}. We shall
focus on the $1$-dimensional discrete version of the problem. In the
introduction we stick to the case of $\Z$ and shall discuss the case
of ${\Z_+:=\Z\cap[0,\infty)}$ in the core of the paper, see Section
\ref{s:N}. The mass of 
the particle is given by $m\geq 0$. For simplicity, we re-normalize
the speed of light and the  reduced Planck constant by $1$.
The  discrete Dirac operator, acting on $\ell^2(\Z,\mathbb{C}^2)$, is
defined by 
\[
D_m:=\left(
\begin{array}{cc}
m & d \\
d^* & -m
\end{array}\right) ,
\]
where $d :={\rm Id} -\tau $ and $\tau$ is the right shift, defined by
$\tau f(n)= f(n+1)$, for all $f\in \ell^2(\Z,
\C)$. The operator $D_m$ is self-adjoint. 
Moreover, notice that  
\[D_m^2= \left(\begin{array}{ll}
\Delta +m^2 & 0
\\
0 & \Delta +m^2
\end{array}\right), \]
where $\Delta f(n):= 2 f(n)- f(n+1)- f(n-1)$. This yields that 
$\sigma(D_m^2)= [m^2, 4+m^2]$. To remove the square above $D_m$, we
define the symmetry $S$ on $\ell^2(\Z,\C)$ by  
$Sf(n):=f(-n)$ and  the unitary operator on $\ell^2(\Z,\C^2)$ 
\begin{align}\label{e:U}
U:=\left(\begin{array}{cc}
0 & \rmi S \\
-\rmi S & 0
\end{array}\right) .
\end{align}
Clearly $U=U^*=U^{-1}$. We have that $U D_mU=-D_m$. We infer that the
spectrum of $D_m$ is purely
absolutely continuous (ac) and that
\[\sigma(D_m)= \sigma_{\rm ac}(D_m)= [-\sqrt{m^2+4},
-m]\cup [m, 
\sqrt{m^2+4}].\]

We shall now perturb the operator by an electrical potential $V =
(V_1, V_2)^t \in \ell^\infty (\Z,\R^2)$. We set 
\begin{align}\label{e:H} 
H:= D_m  + \left(
\begin{array}{cc}
V_1  & 0 \\
0 & V_2 
\end{array}\right).
\end{align}
Here, $V_i$ denotes also the operator of
multiplication by the function $V_i$. Clearly, the essential spectrum of
$H$ is the same as the one of 
$D_m$ if $V$ tends to $0$ at infinity. We turn to more refined
questions. The singular continuous spectrum, quantum transport, and
localization have  been studied before \cite{CaOl, PrOl, COP, PrOl2,
  OlPr,   OlPr2}. The question of the purely ac spectrum seems not to
have been answered before. This is the purpose of our article.   

We recall the following standard result for the Laplacian (the non-relativistic
setting). For completeness we sketch the proof in Section \ref{s:Lap}.

\begin{theorem}\label{t:main0}
Take $V\in \ell^\infty (\Z,\R)$ and $\nu\in \Z_+\setminus\{0\}$ such that:
\begin{align}\label{e:main0}
\left.\begin{array}{c}
\lim_{n\to \pm \infty}V(n)=0,
\\
\\
V_{|_{{\Z_+}}}-\tau^\nu V_{|_{{\Z_+}}}\in\ell^1 ({\Z_+},\R), 
\end{array}\right. 
\end{align}
then the spectrum of $\Delta + V $ is purely absolutely continuous on
$(0, 4)$. 
\end{theorem}

In the case of ${\Z_+}$ and for $\nu=1$, the result has been essentially proved in
\cite{Wei} (in fact in the quoted reference, one focuses only on the
continuous setting).  The proof for the discrete setting can be found
in \cite{DoNe,   Sim}. For $\nu>1$, it seems that it was first done in
\cite{Sto}. Note that  for 
instance, \eqref{e:main0} is satisfied by potentials like $V(n)=
(-1)^n W(n)$, where $W$ is 
decay to $0$.  We refer to \cite{GoNe} and to \cite{KaLa} for recent
results in this direction. 

An amusing and easy remark is the difference between ${\Z_+}$ and $\Z$. In the
latter, it is enough to assume the decay hypothesis on
the right part of the 
potential. This reflects the fact that the particle can always escape
to the right even if the left part of the potential would have given 
some singular continuous spectrum in a half-line setting. 

We now turn to the main result of the paper. For sake of simplicity we
present the case of $\Z$ with electric perturbations. We refer to Section
\ref{s:main} for the main statements, where we deal with a mixture of
magnetic and Witten-like perturbations and also with $\Z^+$.  

\begin{theorem}\label{t:main}
Take $V\in \ell^\infty (\Z,\R^2)$ and $\nu \in \Z_+\setminus\{0\}$ with:
\begin{align*}
\left.\begin{array}{c}
\lim_{n\to \pm \infty}V(n)=0,
\\
\\
V_{|_{{\Z_+}}}-\tau^\nu V_{|_{{\Z_+}}}\in\ell^1 ({\Z_+},\R^2),
\end{array}\right. 
\end{align*}
then the spectrum of $H$ is purely absolutely continuous on
$\left(-\sqrt{m^2+4} , -m\right)\cup \left(m , \sqrt{m^2+4}\right)$. 
\end{theorem}

To study $H$ we reduce the problem to a non-self-adjoint
Laplacian-like operator which depends on the spectral parameter. This
is due to 
a spin-up/down decomposition, see Proposition \ref{p:schur}. 
This idea has  been  efficiently used in the continuous
setting, e.g.,  \cite{DES, BoGo, JeNe} and references therein, and
seems to be new in the discrete setting. Then, 
we adapt the iterative process to the non-self-adjoint Laplacian-like
operator and follow the presentation of \cite{FHS}. We refer to 
\cite{FHS5} for a recent survey about this technique. 

Finally we present the organization of the paper. In  section \ref{s:gene}
we recall general facts about the free discrete Dirac operator and
about hyperbolic geometry. Then in Section \ref{s:main} we present the
main results. Next in Section \ref{s:like}, we reduce the problem to a
kind of Laplacian and adapt the transfer matrices technique. After
that in Section \ref{s:abso}, we prove the main results about
absolutely continuous spectra. Finally we discuss briefly the case of
the Laplacian. 

\noindent {\bf Notation:}
We denote by $\mathcal{B}(X)$ the space of bounded operators acting on
a Banach space $X$. Let $\Z_k:=\Z\cap [k,+\infty [$ for $k\in \Z$ and
$\Z_-:= \Z \setminus {\Z_+}$.  For  
$A,B\subset\C$, we set $A\Subset B$ if ${\rm cl}{A}\subset{\rm
  int}{B}$, where ${\rm cl}$ and ${\rm int}$ stand for closure and
interior, respectively.   

\noindent {\bf Acknowledgments: } We would like to thank Serguei
Denisov, Stanilas Kupin, Christian Remling, and
Hermann Schulz-Baldes for  useful discussions. We would also like to
thank Leonid Golinski for comments on the script. The research is
partially supported by Franco-Ukrainian programm ``Dnipro 2013-14". 

\section{General facts}\label{s:gene}
\subsection{The spectrum of the discrete Dirac operator}
Let $\Z_k:=\Z\cap [k,+\infty [$ for $k\in \Z$ and $\NZ\in\left\{{\Z_k}
  ,\Z\right\}$. We define $d\in\mathcal{B}\left(
  \ell^2(\NZ,\mathbb{C})\right)$ by 
\[d f(n):=f(n)-f(n+1),\]
for all $f\in \ell^2(\NZ,\mathbb{C})$ and  $n\in \NZ$. Clearly $d$ is bounded.
Its adjoint is given by
\[
d^* f(n)=\left\{\begin{array}{ll}
f(n) & \quad\text{ if $\NZ ={\Z_k}$ and $n=k$,}\\
f(n)-f(n-1) & \quad\text{ otherwise,}
\end{array}\right.
\]
for all $f\in \ell^2(\NZ,\C)$ and $n\in \NZ$ .
Now for $m\geq 0$ we define the  discrete Dirac operator on
$\ell^2(\NZ,\mathbb{C}^2)$ by 
\[
D_m^{(\NZ )}:=\left(
\begin{array}{cc}
m & d \\
d^* & -m
\end{array}\right) .
\]
It is easy to see that $D_m^{(\NZ )}$ is self-adjoint.
Let $\Delta^{(\NZ )}$ be the Laplacian on $\ell^2(\NZ,\C)$ defined by
\begin{align}\label{e:DeltaNZ}
\Delta^{(\NZ )}f(n):=\left\{\begin{array}{ll}
f(n)-f(n+1) & \quad\text{ if $\NZ ={\Z_k}$ and $n=k$,}\\
2f(n)-f(n-1)-f(n+1) & \quad\text{ otherwise,}
\end{array}\right.
\end{align}
for all $f\in \ell^2(\NZ,\C)$.
We study first the  discrete Dirac operator on $\Z$, we have
\[
\left( D_m^{(\Z )}\right)^2=\left(\begin{array}{cc}
\Delta^{(\Z )} +m^2 & 0 \\
0 & \Delta^{(\Z )} +m^2
\end{array}\right) .
\]
By Fourier transformation, we see that $\Delta^{(\Z )}$ is
non-negative and that its spectrum is $[0,4]$. Therefore, the
spectrum of $\left( D_m^{(\Z )}\right)^2$ is $[m^2,4+m^2]$. Relying on
\eqref{e:U}, we obtain:

\begin{proposition}
We have
\[
\sigma \left( D_m^{(\Z)}\right)
= \sigma_\ess \left( D_m^{(\Z)}\right)
=\sigma \left( D_m^{({\Z_+})}\right)
= \sigma_\ess \left( D_m^{({\Z_+})}\right)
= \left[-\sqrt{m^2+4},-m\right]\cup\left[m,\sqrt{m^2+4}\right] .
\]
\end{proposition}
\begin{proof}
We have $\left( - D_m^{(\Z)} -\lambda\right)^{-1}=U\left(  D_m^{(\Z)}
  -\lambda\right)^{-1}U$ so $\varphi\left( - D_m^{(\Z)} \right)
=U\varphi\left(  D_m^{(\Z)} \right) U$,  
for all $\varphi $ Borel measurable. Therefore, $\sigma \left(
  D_m^{(\Z)}\right) =
\left[-\sqrt{m^2+4},-m\right]\cup\left[m,\sqrt{m^2+4}\right] $. 
By writing $\Z=\Z_-\cup{\Z_+} $, we see easily that $\sigma_{\rm ess}
(D_m^{(\Z)})=\sigma_{\rm ess }(D_m^{({\Z_+})})$. To conclude, a direct
computation shows that $D_m^{({\Z_+})}$ has no eigenvalue. 
\end{proof}

\subsection{A few words about the Poincar\'e half-plane}
We shall use extensively some properties of the
\emph{Poincar\'e half-plane}. It is defined by: 
\begin{align*}
\demi := \{ x+\rmi y \mid x\in\mathbb{R} , y>0 \}, \mbox{ endowed with
  the metric } 
ds^2=\frac{dx^2+dy^2}{y^2},\end{align*}
Recall that the geodesic distance is given by:
\begin{align}\label{inh}
d_{\demi}(z_1,z_2)=\cosh
^{-1}\left(1+\frac{1}{2}\frac{|z_1-z_2|^2}{\Im z_1 \cdot\Im
    z_2}\right) \leq \frac{|z_1-z_2|}{\sqrt{\Im (z_1)}\sqrt{\Im (z_2)}} .
\end{align}
We turn to the study of (hyperbolic-)contractions. 
\begin{lemma}\label{l:cont}
Given $a,b\in{\rm cl} (\demi)$ and $c>0$, we set:
\begin{align}\label{e:varphiab}
\varphi_{a,b,c}(z):=-\left( a-\left( b+cz\right)^{-1}\right)^{-1} .
\end{align}
It is a contraction of $(\demi, d_\demi)$. 

Moreover, if $a,b\in \demi$, then $\varphi_{a,b,c}$ is a strict
contraction of $(\demi, d_\demi)$. More precisely, we have   
\begin{align}\label{e:cont1}
d_\demi\left(\varphi_{a,b,c}(z_1) , \varphi_{a,b,c}(z_2)\right)\leq
\frac{1}{1+\Im ( a )\cdot \Im ( b )} d_\demi\left(z_1 , z_2\right)
, \end{align}
for all $z_1 ,z_2\in\demi$.
Moreover we have
\[d_\demi\left(\varphi_{a,b,c}(z_1) , \varphi_{a,b,c}(z_2)\right)\leq \frac{\left( 1 +\Im (b )|a |\right)^2}{\left( \Im (b )\cdot\Im (a )\right)^2} ,
 \]
for all $z_1 ,z_2\in\demi$.
\end{lemma}
\begin{proof}
First, since $z\mapsto c z$ is a hyperbolic isometry, it is enough to
consider $\varphi_{a,b}:= \varphi_{a,b,1}$. Then, using that $z\mapsto
z+w$ is a contraction when $w\in {\rm cl }(\demi)$ and that $z\mapsto
-1/z$ is an isometry, we obtain that $\varphi_{a,b}$ is a
contraction. 

We turn to the second statement. A direct computation yields that if
$h\in\demi$ then 
\begin{align}\label{e:im2}
-\left( h +\demi \right)^{-1} & = B_{|.|}\left( \frac{\rmi}{2\Im(h)} ,\frac{1}{2\Im(h)}\right)  \subset\demi  .
\end{align}
Given $C>0$, as in the proof of \cite{FHS}[Proposition 2.1], if $z_1,z_2\in\demi$ with $\min\left(|z_1|,|z_2|\right) \leq C$ note that
\[
d_{\demi}(z_1+a ,z_2+a)\leq \frac{C}{C+\Im ( a )}d_{\demi}(z_1 ,z_2) .
\]
Since $z\mapsto -z^{-1}$ is an isometry of $\demi$ and $z\mapsto -(b+z)^{-1}$ is a contraction of $\demi$, we use \eqref{e:im2} for $h=b$ to have that
\begin{align*}
d_\demi\left(\varphi_{a,b}(z_1) ,\varphi_{a,b}(z_2)\right)
&=d_\demi\left( a-\left( b+z_1\right)^{-1} , a-\left(
    b+z_2\right)^{-1} \right) \\ 
& \leq \frac{\left(\Im ( b )\right)^{-1}}{\left(\Im ( b
    )\right)^{-1}+\Im ( a )}d_\demi\left( -\left(
    b+z_1\right)^{-1} , -\left( b+z_2\right)^{-1} \right) \\ 
& \leq 
\frac{1}{1+\Im ( a )\cdot \Im ( b )}d_\demi\left(z_1 , z_2\right) ,
\end{align*}
for all $z_1,z_2\in\demi$.
So we obtain \eqref{e:cont1}.
By \eqref{e:im2} for $h=a$ we know that $-\left( a +\demi\right)^{-1}$
is an Euclidean ball of diameter $\Im(a )^{-1}$, hence 
\[
\left| \varphi_{a,b}(z_1)-\varphi_{a,b}(z_2)\right|\leq \left(\Im(a )\right)^{-1} 
\]
for all $z_1,z_2\in\demi$.
Given $C>0$, if $z\in\demi$ with $|z|\leq C$ we have
\[
\Im \left(-\frac{1}{a + z }\right)\geq \frac{\Im(a )}{\left( C +|a |\right)^2} .
\]
So if $z\in\demi$, by \eqref{e:im2} for $h=b$, $|-(b+z)^{-1}|\leq \left(\Im (b)\right)^{-1}$, so
\[
\Im \left(\varphi_{a,b}(z)\right)\geq \frac{\Im(a )}{\left( \left(\Im ( b )\right)^{-1} +|a |\right)^2} .
\]
So we obtain that
\begin{align*}
d_\demi\left(\varphi_{a,b}(z_1) ,\varphi_{a,b}(z_2)\right)
 = & \cosh^{-1}\left(1+\frac{1}{2}\frac{|\varphi_{a,b}(z_1) - \varphi_{a,b}(z_2)|^2}{\Im (\varphi_{a,b}(z_1))\Im (\varphi_{a,b}(z_2))}\right)
 \leq \frac{|\varphi_{a,b}(z_1) - \varphi_{a,b}(z_2)|}{\sqrt{\Im (\varphi_{a,b}(z_1))}\sqrt{\Im (\varphi_{a,b}(z_2))}}  \\
 \leq  & \frac{\left( \left(\Im ( b )\right)^{-1} +|a |\right)^2}{\left(\Im (a )\right)^2} 
 =\frac{\left( 1 +\Im (b )|a |\right)^2}{\left( \Im (b )\cdot\Im (a )\right)^2}   ,
\end{align*}
for all $z_1,z_2\in\demi$.
\end{proof}

We shall also need the following technical lemmata.

\begin{lemma}\label{l:phiim}
Suppose that $a,b\in{\rm cl} (\demi)$, $c>0$, and $z\in \demi$. We
have:
\[
d_{\demi}(\varphi_{a,b,c}(z),\rmi)\leq \left(\frac{(|b|+c |z|)^2 }{c
    \Im z}+1\right)\frac{\left( |b|+c|z| \right)\left( |a|+\left(
      c\Im\left( z\right)\right)^{-1} \right)}{\sqrt{c\Im\left(
      z\right)}} . 
\]
\end{lemma}
\begin{proof}
First we have:
\begin{align*}
|\varphi_{a,b,c}(z)|
& =
\left|\frac{1}{a-\left( b+cz\right)^{-1}}\right|
\leq
\frac{1}{\Im \left(a-\left( b+cz\right)^{-1}\right)}
\leq
\frac{1}{\Im \left(-\left( b+cz\right)^{-1}\right)}
=
\frac{\left| b+cz\right|^2}{\Im \left(b+cz\right)}
\leq
\frac{\left( |b|+c|z|\right)^2}{c\Im \left(z\right)} .
\end{align*}
Then we obtain:
\begin{align*}
\Im\left(\varphi_{a,b,c}(z) \right)
& =
\Im\left(-\frac{1}{a-\left( b+cz\right)^{-1}}\right)
=
\frac{\Im\left(a-\left( b+cz\right)^{-1}\right)}{\left| a-\left(
      b+cz\right)^{-1} \right|^2} 
\geq
\frac{\Im\left(-\left( b+cz\right)^{-1}\right)}{\left( |a|+\left|
      b+cz\right|^{-1} \right)^2} 
\\ & \geq
\frac{\Im\left( b+cz\right)}{\left| b+cz\right|^2}\left( |a|+\Im\left(
    b+cz\right)^{-1} \right)^{-2} 
\geq \frac{c\Im\left( z\right)}{\left( |b|+c|z| \right)^2\left(
    |a|+\left( c\Im\left( z\right)\right)^{-1} \right)^{2}} . 
\end{align*}
Finally with \eqref{inh} we infer:
\[
d_{\demi}(\varphi_{a,b,c}(z),\rmi )\leq \frac{|\varphi_{a,b,c}(z)-\rmi
  |}{\sqrt{\Im\left( \varphi_{a,b,c}(z)\right)}}\leq
\frac{|\varphi_{a,b,c}(z)|+1 }{\sqrt{\Im\left(
      \varphi_{a,b,c}(z)\right)}},
\]
which yields the result.
\end{proof}

\begin{lemma}\label{l:pol}
For all $n\in\Z_1$ there exist
\[
A_n,B_n,C_n,D_n\in\R [X_1,Y_1,Z_1,\ldots ,X_n,Y_n,Z_n] ,
\]
such that  $C_n(\omega )+D_n(\omega )\zeta\not =0$ and
\[
\varphi_{x_1,y_1,z_1}\circ\cdots\circ\varphi_{x_n,y_n,z_n}(\zeta)=-\frac{A_n(\omega )+B_n(\omega )\zeta}{C_n(\omega )+D_n(\omega )\zeta},
\]
for all $\omega :=(x_1,y_1,z_1,\ldots ,x_n,y_n,z_n)\in \left( {\rm
    cl}(\demi )^2\times\R_+^*\right)^n$ and $\zeta\in\demi$. 
\end{lemma}

We point out that the continuity of $A_n, B_n, C_n,$ and $D_n$ with
respect to the coefficients will be crucial in Proposition
\ref{p:born}. The proof is straightforward. We give it for 
completeness. 

\begin{proof}
We prove the result by induction. Let $n=1$, we have
\[
\varphi_{x,y,z}(\zeta )=-\left( x-\left( y+z\zeta\right)^{-1}\right)^{-1}=-\frac{y+z\zeta}{xy-1+xz\zeta}=-\frac{A_1(x,y,z)+B_1(x,y,z)\zeta}{C_1(x,y,z)+D_1(x,y,z)\zeta}
\]
with
\[
A_1(x,y,z):=y
,~
B_1(x,y,z):=z
,~
C_1(x,y,z):=xy-1
,\text{ and }
D_1(x,y,z):=xz .
\]
Moreover
\[
C_1(x,y,z)+D_1(x,y,z)\zeta =xy-1+xz\zeta= (y+z\zeta )\left( x-\frac{1}{y+z\zeta}\right)\not =0
\]
for all $(x,y,z)\in  {\rm cl}(\demi )^2\times\R_+^*$ and $\zeta\in\demi$ because this is the product of two elements of $\demi$.

Now suppose that we have proved the existence of $A_n,B_n,C_n,D_n$ and
prove the existence of $A_{n+1}$, $B_{n+1}$, $C_{n+1}$, and
$D_{n+1}$. Let $\omega:= (x_1,y_1,z_1,\ldots
,x_{n+1},y_{n+1},z_{n+1})\in \left( {\rm cl}(\demi
  )^2\times\R_+^*\right)^{n+1}$ and $\zeta\in\demi$. Let
$\tilde{\omega}:=(x_1,y_1,z_1,\ldots ,x_n,y_n,z_n)$, we have 
\begin{align*}
\varphi_{x_1,y_1,z_1}\circ\cdots\circ \varphi_{x_{n+1},y_{n+1},z_{n+1}}(\zeta)
& =
\varphi_{x_1,y_1,z_1}\circ\cdots\circ \varphi_{x_n,y_n,z_n}(\varphi_{x_{n+1},y_{n+1},z_{n+1}}(\zeta))
\\ & \hspace*{-4cm} =
-\frac{\displaystyle A_n(\tilde{\omega} )-B_n(\tilde{\omega} )\frac{A_1(x_{n+1},y_{n+1},z_{n+1})+B_1(x_{n+1},y_{n+1},z_{n+1})\zeta}{C_1(x_{n+1},y_{n+1},z_{n+1})+D_1(x_{n+1},y_{n+1},z_{n+1})\zeta}}{\displaystyle C_n(\tilde{\omega} )-D_n(\tilde{\omega} )\frac{A_1(x_{n+1},y_{n+1},z_{n+1})+B_1(x_{n+1},y_{n+1},z_{n+1})\zeta}{C_1(x_{n+1},y_{n+1},z_{n+1})+D_1(x_{n+1},y_{n+1},z_{n+1})\zeta}}
 =
 -\frac{A_{n+1}(\omega )+B_{n+1}(\omega )\zeta}{C_{n+1}(\omega )+D_{n+1}(\omega )\zeta} ,
\end{align*}
with
\begin{align*}
& A_{n+1}(\omega ):= A_n(\tilde{\omega} )C_1(x_{n+1},y_{n+1},z_{n+1})-B_n(\tilde{\omega} )A_1(x_{n+1},y_{n+1},z_{n+1}) ,
\\ &
B_{n+1}(\omega ):= A_n(\tilde{\omega} )D_1(x_{n+1},y_{n+1},z_{n+1})-B_n(\tilde{\omega} )B_1(x_{n+1},y_{n+1},z_{n+1}) ,
\\ &
C_{n+1}(\omega ):= C_n(\tilde{\omega} )C_1(x_{n+1},y_{n+1},z_{n+1})-D_n(\tilde{\omega} )A_1(x_{n+1},y_{n+1},z_{n+1}) ,
\\ &
D_{n+1}(\omega ):= C_n(\tilde{\omega} )D_1(x_{n+1},y_{n+1},z_{n+1})-D_n(\tilde{\omega} )B_1(x_{n+1},y_{n+1},z_{n+1}) .
\end{align*}
Finally since $\varphi_{x_{n+1},y_{n+1},z_{n+1}}$ is a contraction by
Lemma \ref{l:cont}, we have that
$\varphi_{x_{n+1},y_{n+1},z_{n+1}}(\zeta)\in \demi$ and 
\begin{align*}
C_{n+1}(\omega )+D_{n+1}(\omega )\zeta = & \left( C_n(\tilde{\omega } )+D_n(\tilde{\omega } )\varphi_{x_{n+1},y_{n+1},z_{n+1}}(\zeta) \right)
\\ & \quad \quad \quad 
\times \left( C_1(x_{n+1},y_{n+1},z_{n+1})+D_1(x_{n+1},y_{n+1},z_{n+1})\zeta \right) \not=0,
\end{align*}
by induction. This finishes the proof. 
\end{proof}

\section{Main results}\label{s:main}
In this paper, we study a Jacobi-like version of $D^{(\NZ)}_m+V $, given $V
= (V_1, V_2)^t \in \ell^\infty (\NZ,\R^2)$ and $W = (W_1, W_2)^t \in
\ell^\infty (\NZ,\C^2)$, where $\NZ\in\left\{{\Z_+}
  ,\Z\right\}$. Let  
\begin{align}\label{e:vraiH}
H^{(\NZ )}_{m,V,W}:=\left(\begin{array}{cc}
m+V_1  & \td \\
\td^* & -m+V_2 
\end{array}\right) ,
\end{align}
where $\td:=d+W_1 +W_2 \tau $ with $\tau f(n):=f(n+1)$ and $d={\rm Id}
- \tau$. Its study is motivated by the fact that it is an intermediate
model before the study of discrete Dirac operator in weighted spaces,
which would be the treated in another place. The perturbation in $W_i$
is a magnetic/Witten-like perturbation. When the perturbation is
purely magnetic, one could simply gauge away the $W$. However, in some
situations, even if one could remove the magnetic perturbation, it is 
sometimes important to be able to lead the analysis till its very end
in order to be able to choose the best [sequence of] gauge[s]. Having
$W_i$ real is motivated by the fact that this corresponds to a dicrete
relativistic analog of the Witten Laplacian.  


We start with the main result on ${\Z_+}$. It will be proved in Section
\ref{s:N}. 

\begin{theorem}\label{t:spn}
Take $V\in \ell^\infty ({\Z_+},\R^2)$ and $W\in \ell^\infty ({\Z_+},\C^2)$ with:
\begin{align}\label{e:spn1}
\lim_{n\to\infty}V(n)=\lim_{n\to\infty}W(n)=0,
\end{align}
then $\sigma_\ess (H^{({\Z_+} )}_{m,V,W})=\left[ -\sqrt{m^2+4} ,
  -m\right]\cup \left[ m , \sqrt{m^2+4}\right]$. Assuming also that
there exist $\nu_1, \nu_2\in\Z_+\setminus\{0\}$ such that 
\begin{align}\label{e:spn2}
\left.\begin{array}{c}
W_1(n)\not =-1 \mbox{ and } W_2(n)\not=1, \mbox{ for all } n\in{\Z_+},
\\
\\
V-\tau^{\nu_1} V\in\ell^1 ({\Z_+},\R^2) \mbox{ and } W-\tau^{\nu_2}
W\in\ell^1 ({\Z_+},\C^2), 
\end{array}\right. 
\end{align}
then the spectrum of $H^{({\Z_+} )}_{m,V,W}$ is purely absolutely
continuous on $\left(-\sqrt{m^2+4} , -m\right)\cup \left(m ,
  \sqrt{m^2+4}\right)$. 
\end{theorem}

We discuss briefly the necessity of the first line of
\eqref{e:spn2}. 

\begin{remark}
Assume that there is $n_0$ such that $W_1(n_0+1)=-1$ and $W_2(n_0+1)=1$, then
the operator is a direct sum of $H^{(\Z_{n_0+1} )}_{m,V,W}$ and of a
finite matrix. Therefore, it is easy to construct embedded eigenvalues
for this operator and this is an obstruction for the result of the
theorem. For instance suppose that $W_1(k)=W_2(k)=0$ for all $k\neq
n_0$ and $V_1=V_2= 0$. Recalling that in a finite dimensional Hilbert space
$X$ the spectra of 
$MN$ and of $NM$ are equal, for $M,N\in \Bc(X)$, we have
that the point spectrum of $\left(H^{(\Z_{+} )}_{m,0,W}\right)^2$ is exaclty
the one of 
\begin{align*}
m^2{\rm Id}_{\{1,2,\ldots n_0\}}+
\left(\begin{array}{rrrrrr}
1 & -1 & 0 & \ldots & 0& 0
\\
-1 & 2 & -1 & \ldots & 0 & 0
\\
0 & -1 & 2 & \ldots & 0 & 0
\\
0 & 0 & \ddots & \ddots & \vdots & \vdots
\\
0 & 0 & 0 &  \ldots & 2 & -1
\\
0 & 0 & 0 &  \ldots & -1 & 1
\end{array}\right). 
\end{align*}
Therefore the point spectrum of $H^{(\Z_{+} )}_{m,0,W}$ is contained
in $\left\{\pm\sqrt{m^2+ 4 \sin^2(k\pi/2n_0)}, k=0, \ldots,
  n_0-1\right\}$ with total multiplicity $2n_0$.
\end{remark}

We turn to the case of $\Z$, we have this important symmetry of charge:
\begin{align}\label{e:sym}
UH^{(\Z )}_{m,V,W}U=-H^{(\Z )}_{m,\left( -SV_2 , -SV_1 \right)^t
  ,\left( S\overline{W_1} , \tau S\overline{W_2}\right)^t}, 
\end{align}
where $U$ is given by \eqref{e:U}. 

\begin{theorem}\label{t:spz}
Take $V\in \ell^\infty (\Z,\R^2)$ and $W\in \ell^\infty (\Z,\C^2)$ with:
\[
\lim_{n\to\pm\infty}V(n)=\lim_{n\to\pm\infty}W(n)=0,
\]
then $\sigma_\ess (H^{(\Z )}_{m,V,W})=\left[ -\sqrt{m^2+4} ,
  -m\right]\cup \left[ m , \sqrt{m^2+4}\right]$. Assuming also that
there exists $\nu_1, \nu_2\in\Z_+\setminus\{0\}$ such that 
\begin{align}\label{e:spz2}
\left.\begin{array}{c}
W_1(n)\not =-1 \mbox{ and } W_2(n)\not=1, \mbox{ for all } n\in\Z,
\\
\\
V_{|_{{\Z_+}}}-\tau^{\nu_1} V_{|_{{\Z_+}}}\in\ell^1 ({\Z_+},\R^2) \mbox{ and }
W_{|_{{\Z_+}}}-\tau^{\nu_2} W_{|_{{\Z_+}}}\in\ell^1 ({\Z_+},\C^2),
\end{array}\right. 
\end{align}
or alternatively
\begin{align}\label{e:spz2'}
 V_{|_{\Z_-}}-\tau^{\nu_1} V_{|_{\Z_-}}\in\ell^1 (\Z_-,\R^2) \mbox{ and
 } W_{|_{\Z_-}}-\tau^{\nu_2} W_{|_{\Z_-}}\in\ell^1 (\Z_-,\C^2),
\end{align}
then the spectrum of $H^{(\Z )}_{m,V,W}$ is purely absolutely
continuous on $\left(-\sqrt{m^2+4} , -m\right)\cup \left(m ,
  \sqrt{m^2+4}\right)$.
\end{theorem}

\begin{remark}
Note that the second line of \eqref{e:spz2} and \eqref{e:spz2'} 
are equivalent  by using the transformation $U$.
\end{remark}

\section{A Laplacian-like approach}\label{s:like}
\subsection{Another form for the resolvent}
The objective is to reduce the analysis of the operator $H^{(\NZ
  )}_{m,V,W}$ on $\ell^2(\NZ ,\C^2)$ to the one of two operators which
are similar to a Laplacian.   
Take $\lambda\in\demi$, $V = (V_1, V_2)^t \in \ell^\infty 
(\NZ,\R^2)$, and $W = (W_1, W_2)^t \in \ell^\infty (\NZ,\C^2)$, we
define  
\[\Delta^{(\NZ )}_{1,\lambda
  ,V,W}:=\td\frac{1}{\lambda+m-V_2 }\td^*-(\lambda-m-V_1  )\] 
and
\begin{align}\label{e:delta2}
\Delta^{(\NZ )}_{2,\lambda
  ,V,W}:=\td^*\frac{1}{\lambda-m-V_1 }\td-(\lambda+m-V_2 ).
\end{align} 
Note that to lighten the notation we have dropped the dependency on
$m$. However we keep the one in $\lambda$. We first check their invertibility. 

\begin{proposition}\label{p:inv}
Let $\lambda\in\demi$, $V = (V_1, V_2)^t \in \ell^\infty (\NZ,\R^2)$,
and $W = (W_1, W_2)^t \in \ell^\infty (\NZ,\C^2)$, then $\Delta^{(\NZ
  )}_{1,\lambda ,V,W}$ and $\Delta^{(\NZ )}_{2,\lambda ,V,W}$ are
invertible. 
\end{proposition}
\begin{proof}
For $b\in \mathcal{B}\left( \ell^2(\NZ,\mathbb{C})\right)$, $X,Y\in
\ell^\infty (\NZ,\mathbb{R})$ and $\mu\in\demi$ let $A_{\mu , b,X,Y}
:=b^*\left(\mu -X \right)^{-1}b+Y $, then 
\begin{align}\label{e:im}
\Im \langle f , A_{\mu ,b,X,Y} f\rangle =-\Im (\mu )\left\|
  \frac{1}{|\mu -X |}bf\right\|^2 \leq 0, 
\end{align}
for $f\in \ell^2({\NZ},\mathbb{C})$.
With the Numerical Range Theorem (e.g., \cite[Lemma B.1]{BoGo}) we
derive that we have $\demi \subset\rho (A_{\mu ,b,X,Y})$, the resolvent set of
$A_{\mu ,b,X,Y}$. Since
\begin{align*}
\Delta^{(\NZ )}_{1,\lambda ,V,W}=A_{\lambda
  ,\td^*,V_2-m,V_1+m}-\lambda \quad\text{ and }\quad \Delta^{(\NZ
  )}_{2,\lambda ,V,W}=A_{\lambda ,\td,V_1+m,V_2-m}-\lambda , 
\end{align*}
we get $\Delta^{(\NZ )}_{1,\lambda ,V}$ and $\Delta^{(\NZ
  )}_{2,\lambda ,V}$ are invertible. 
\end{proof}

We give a kind of Schur's Lemma, so as to compute the inverse of the
Dirac operator, see also \cite{DES}, \cite{BoGo}, and \cite{JeNe} for
some applications in the continuous setting.

\begin{proposition}\label{p:schur}
Let $\lambda\in\demi$, $V = (V_1, V_2)^t \in \ell^\infty (\NZ,\R^2)$, and $W = (W_1, W_2)^t \in \ell^\infty (\NZ,\C^2)$. Then :
\[ \left( H^{(\NZ )}_{m,V,W}-\lambda \right)^{-1}=\left(\begin{array}{cc}
(\Delta^{(\NZ )}_{1,\lambda ,V,W})^{-1} & 0 \\
0 & (\Delta^{(\NZ )}_{2,\lambda ,V,W})^{-1}
\end{array}\right)\left(\begin{array}{cc}
1 & \td\frac{1}{\lambda +m -V_2 } \\
\td^*\frac{1}{\lambda -m -V_1 } & 1
\end{array}\right). \]
\end{proposition}
\begin{proof} We set $(H^{(\NZ )}_{m,V,W}-\lambda )f=g$. This gives:
\begin{align*}
&\left\{\begin{array}{ccc}
\displaystyle (V_1 -\lambda +m)f_1+\td f_2 &=& g_1 \\
\displaystyle \td^*f_1 + (V_2 -\lambda -m)f_2 &=& g_2 \\
\end{array}\right. , \quad\quad  \left\{\begin{array}{ccc}
f_1 &=& \displaystyle\frac{1}{V_1 -\lambda +m}(g_1-\td f_2)\\
f_2 &=& \displaystyle \frac{1}{V_2 -\lambda -m}(g_2-\td^*f_1) \\
\end{array}\right.\\
&\left\{\begin{array}{ccc}
f_1 &=& \displaystyle \frac{1}{V_1 -\lambda +m}\left(g_1-\td\frac{1}{V_2 -\lambda -m}(g_2-\td^*f_1)\right)\\
f_2 &=& \displaystyle \frac{1}{V_2 -\lambda -m}\left(g_2-\td^*\frac{1}{V_1 -\lambda +m}(g_1-\td f_2)\right) \\
\end{array}\right. \\
&\left\{\begin{array}{ccc}
\displaystyle \left(\td\frac{1}{\lambda +m -V_2 }\td^*-(\lambda -m-V_1 )\right)f_1 &=&\displaystyle  g_1 +\td\frac{1}{\lambda +m -V_2 }g_2\\
\displaystyle\left(\td^*\frac{1}{\lambda -m -V_1 }\td-(\lambda
  +m-V_2 )\right)f_2 &=& \displaystyle \td^*\frac{1}{\lambda -m
  -V_1 }g_1 +g_2\\ 
\end{array}\right. 
\end{align*}
Since $\Delta^{(\NZ )}_{1,\lambda ,V,W}$ and $\Delta^{(\NZ
  )}_{2,\lambda ,V,W}$ are invertible, we obtain the result. 
\end{proof}

\subsection{Study of the truncated operator}

Note that in \eqref{e:delta2}, if we forget about the terms in
$\lambda$, $V$, and  $W$, we obtain a Laplacian on ${\Z_+}$. Moreover,
motivated by the results of Sections \ref{s:N} and \ref{s:Z}, it is 
enough to focus the analysis on the study of $\Delta^{({\Z_+}
  )}_{2,\lambda ,V,W}$. Therefore, we stress that we will not study
$\Delta^{({\Z_+} )}_{1,\lambda ,V,W}$ at all. In fact, the
latter leads to some technical complications and is less natural,
i.e., it is not a direct analogue of the Laplacian on ${\Z_+}$.

As in \cite{FHS}, we reduce the problem to $\Z_k$ for $k\in {\Z_+}$. 
We define the truncated operator
$\td=\td(n)  \in\mathcal{B}\left(\ell^2(\Z_n,\C)\right)$ by 
\[
\td = d+W_1+W_2 \tau.
\]
Now we define $\Delta^{(n)}_{\lambda ,V,W}\in\mathcal{B}\left(\ell^2(\Z_n,\C)\right)$ by
\[
\Delta^{(n)}_{\lambda ,V,W}:=
\td ^*\frac{1}{\lambda-m-V_{1|_{\Z_n}} }\td-(\lambda+m-V_{2|_{\Z_n}}  )
\]
Again to lighten notation, we drop the dependency in $m$. We point out that:
\[\Delta^{(0)}_{\lambda ,V,W}=\Delta^{({\Z_+} )}_{2,\lambda ,V,W} .
\]

\begin{proposition}\label{p:inv2}
Let $\lambda\in\demi$, $V = (V_1, V_2)^t \in \ell^\infty ({\Z_+},\R^2)$, and $W = (W_1, W_2)^t \in \ell^\infty ({\Z_+},\C^2)$, then $\Delta^{(n )}_{\lambda ,V,W}$ is invertible for all $n\in{\Z_+}$.
\end{proposition}
\begin{proof}
This is essentially the same proof as for Proposition \ref{p:inv}.
\end{proof}

We study the related Green function:
\[
\alpha_n:=\left\langle\delta_n , \left(\Delta^{(n)}_{\lambda ,V,W}\right)^{-1}\delta_n\right\rangle,
\]
where $\delta_n(m):=1$ if and only if $n=m$ and $0$ otherwise.
The objective is to bound $\alpha_0$ independently of $\lambda$.
We give the first property of $\alpha_n$.

\begin{proposition}\label{p:dem}
Take $\lambda\in\demi$,  $V\in \ell^\infty ({\Z_+},\R^2)$, and $W\in
\ell^\infty ({\Z_+},\C^2)$ then 
\[
\alpha_n=\left\langle\delta_n , \left(\Delta^{(n)}_{\lambda
      ,V,W}\right)^{-1}\delta_n\right\rangle\in\demi, 
\]
for all $n\in{\Z_+}$.
\end{proposition}
\begin{proof}
We have
\begin{align*}
\Im (\alpha_n)= 
\Im (\lambda )\left(\left\|
    \frac{1}{|\lambda-m-V_{1|_{\Z_n}} |}\td\left(
        \Delta^{(n)}_{\lambda           ,V,W}  
\right)^{-1}\delta_n \right\|^2 + \left\|\left(
\Delta^{(n)}_{\lambda ,V,W}
\right)^{-1}\delta_n\right\|^2 \right) .
\end{align*}
So $\Im (\alpha_n) >0$ because $\lambda\in\demi$ and $\left(
\Delta^{(n)}_{\lambda ,V,W}
\right)^{-1}\delta_n\not =0$.
\end{proof}

We follow the strategy of \cite{FHS} and express $\alpha_n$ with the help of $\alpha_{n+1}$. The aim is to use a fixed point argument in order to recover some bounds on $\alpha_0$.

\begin{proposition}\label{p:rec}
Take $\lambda\in\demi$,  $V\in \ell^\infty ({\Z_+},\R^2)$,  $W\in \ell^\infty ({\Z_+},\C^2)$ with $W_1(k)\not =-1$ for all $k\in{\Z_+}$, and $n\in{\Z_+}$. By setting
\begin{align}
\nonumber
\Phi_n(z) & := \varphi_{a_n,b_n, c_n}(z) = - \left(
a_n-\left(
b_n +c_nz
\right)^{-1}
\right)^{-1} \\
\nonumber
a_n & := \lambda +m -V_2(n)\in\demi \\
\label{e:bn}
b_n & :=(\lambda -m -V_1(n))|1+W_1(n)|^{-2}\in\demi \\
\nonumber
c_n & := \left|\frac{1-W_2(n)}{1+W_1(n)}\right|^2\in\R_+,
\end{align}
we obtain
$ \alpha_n=\Phi_n(\alpha_{n+1}) $.
\end{proposition}
\begin{proof}
We define in $\ell^2(\Z_n,\C)$ and in $\ell^2(\Z_{n+1},\C)$
\[
f:=\left(\Delta^{(n)}_{\lambda ,V,W}\right)^{-1}\delta_n
\quad\text{ and }\quad
g:=\left(\Delta^{(n+1)}_{\lambda ,V,W}\right)^{-1}\delta_{n+1} ,
\]
respectively.
Clearly $\alpha_n=f(n)$ and $\alpha_{n+1}=g(n+1)$. By definition $f$
is the unique solution in $\ell^2(\Z_n,\C)$ of
$\Delta^{(n)}_{\lambda ,V,W}f=\delta_n$, i.e., 
\begin{align}
\label{rec}
\frac{1+\overline{W_1}(k)}{\lambda -m -V_1(k)}\left((1+W_1(k))f(k)+(-1+W_2(k))f(k+1)\right)\hspace*{5 cm} \\
\notag 
+\frac{1-\overline{W_2}(k-1)}{\lambda -m -V_1(k-1)}\left((1-W_2(k-1))f(k)+(-1-W_1(k-1))f(k-1)\right)
 -(\lambda +m-V_2(k))f(k)  =0 ,
\end{align}
for all $k\geq n+1$ and 
\begin{align}\label{init}
& \frac{1+\overline{W_1}(n)}{\lambda -m -V_1(n)}\left((1+W_1(n))f(n)+(-1+W_2(n))f(n+1)\right) -(\lambda +m-V_2(n))f(n)  =1 .
\end{align}
We see that $f_{|_{\Z_{n+1}}}$ is solution of \eqref{rec} for all
$k\geq n+2$ and  
\begin{align*}
&\quad \frac{1+\overline{W_1}(n+1)}{\lambda -m -V_1(n+1)}\left((1+W_1(n+1))f(n+1)+(-1+W_2(n+1))f(n+2)\right) \\
\notag & \quad \quad \quad \quad \quad \quad 
-(\lambda +m-V_2(n+1))f(n+1) \\
\notag & \quad \quad \quad \quad \quad \quad \quad \quad \quad \quad \quad \quad =
\frac{1-\overline{W_2}(n)}{\lambda -m -V_1(n)}\left((1+W_1(n))f(n)+(-1+W_2(n))f(n+1)\right) .
\end{align*}
So we obtain that
\[
\Delta^{(n+1)}_{\lambda ,V,W}f_{|_{\Z_{n+1}}}=\frac{1-\overline{W_2}(n)}{\lambda -m -V_1(n)}\left((1+W_1(n))f(n)+(-1+W_2(n))f(n+1)\right)\delta_{n+1} .
\]
Because $\Delta^{(n+1)}_{\lambda ,V,W}g=\delta_{n+1}$ we have
\[
\Delta^{(n+1)}_{\lambda ,V,W}f_{|_{\Z_{n+1}}}=\frac{1-\overline{W_2}(n)}{\lambda -m -V_1(n)}\left((1+W_1(n))f(n)+(-1+W_2(n))f(n+1)\right)\Delta^{(n+1)}_{\lambda ,V,W}g
\]
But $\Delta^{(n+1)}_{\lambda ,V,W}$ is invertible, so
\[
f_{|_{\Z_{n+1}}}=\frac{1-\overline{W_2}(n)}{\lambda -m -V_1(n)}\left((1+W_1(n))f(n)+(-1+W_2(n))f(n+1)\right)g .
\]
Note that
\[
f(n+1)=\frac{1-\overline{W_2}(n)}{\lambda -m -V_1(n)}\left((1+W_1(n))f(n)+(-1+W_2(n))f(n+1)\right)g(n+1) .
\]
Straightforwardly, using \eqref{init} we conclude that $\alpha_n =
f(n) = \Phi_n(g(n+1)) =\Phi_n(\alpha_{n+1}) $. 
\end{proof}

\subsection{An iterative process}
The key to the process relies on the fact that $\Phi_n$ is a strict contraction.

\begin{proposition}\label{p:cont}
Given $\lambda\in\demi$, $n\in{\Z_+}$, $V\in \ell^\infty ({\Z_+},\C^2)$, and $W \in \ell^\infty ({\Z_+},\C^2)$ with $W_1(n)\not =-1$ and $W_2(n)\not=1$.
Then $\Phi_n$ is a strict contraction. More precisely, we have 
\begin{align}\label{e:cont2}
d_\demi\left(\Phi_n(z_1) , \Phi_n(z_2)\right)\leq \frac{1}{1+\left(\Im
    ( \lambda )\right)^2 \left( 1+\|W_1\|_\infty\right)^{-2}}
d_\demi\left(z_1 , z_2\right) 
, \end{align}
for all $z_1 ,z_2\in\demi$ and $n\in{\Z_+}$.
Moreover we obtain
\[d_\demi\left(\Phi_n(z_1) , \Phi_n(z_2)\right)\leq
\frac{\left( \left(1+\|W_1\|_\infty\right)^{2} +\Im(\lambda)\left(|\lambda | + m + \left\|V_2\right\|_\infty\right)\right)^2}{\left(\Im(\lambda)\right)^4} , 
 \]
for all $z_1 ,z_2\in\demi$ and $n\in{\Z_+}$.
\end{proposition}
\begin{proof}
Using Lemma \ref{l:cont}, we obtain that $\Phi_n=
\varphi_{a_n,b_n, c_n}$ is a strict  contraction. More
precisely, we get 
\begin{align*}
d_\demi\left(\Phi_n(z_1) , \Phi_n(z_2)\right)
\leq &
\frac{1}{1+\Im (a_{ n})\Im (b_{ n})} d_\demi\left(z_1 , z_2\right) 
\leq 
\frac{1}{1+\left(\Im ( \lambda )\right)^2 \left(
    1+\|W_1\|_\infty\right)^{-2}} d_\demi\left(z_1 , z_2\right), 
\end{align*}
for all $z_1,z_2\in\demi$, and
\begin{align*}
d_\demi\left(\Phi_n(z_1) , \Phi_n(z_2)\right)
\leq &
\left(\frac{ \frac{1}{\Im (b_{ n} )} +|a_{ n} |}{\Im (a_{ n} )}\right)^2 \leq 
\frac{\left( \left(1+\|W_1\|_\infty\right)^{2} +\Im(\lambda)\left(|\lambda | + m + \left\|V_2\right\|_\infty\right)\right)^2}{\left(\Im(\lambda)\right)^4},  
\end{align*}
for all $z_1,z_2\in\demi$.
\end{proof}

Now we have an asymptotic property. That is an analogue of
\cite{FHS}[Theorem 2.3]. It relies strongly on the fact that $\Phi_n$
is a strict contraction. 

\begin{corollary}\label{c:cvg}
Take  $V\in \ell^\infty ({\Z_+},\R^2)$, and $W\in \ell^\infty ({\Z_+},\C^2)$
with $W_1(n)\not =-1$ and $W_2(n)\not=1$ for all $n\in{\Z_+}$. Then for
all $\lambda\in \demi$ and $(\zeta_n)_n\in\demi^{\Z_+}$ we have 
\[
\dhlim_{n\to\infty}\Phi_0\circ\cdots\circ \Phi_n(\zeta_n)=\alpha_0,
\]
i.e., $\lim_{n \to\infty} d_\demi(\Phi_0\circ\cdots\circ
\Phi_n(\zeta_n),\alpha_0)=0$.
\end{corollary}
\begin{proof}
With Proposition \ref{p:dem}, for all $n\in{\Z_+}$ we have
$\alpha_{n}\in\demi$. With Proposition
\ref{p:cont} there exist $\delta\in (0,1)$ and $\eta >0$ such that
\[
d_\demi\left(\Phi_n(z_1) , \Phi_n(z_2)\right)\leq \min\left(\delta
  d_\demi\left(z_1 , z_2\right) , \eta\right) 
, \] 
for all $n\in{\Z_+}$ and $z_1,z_2\in\demi$.
So, using that $\alpha_n=\Phi_n (\alpha_{n+1})$ for all $n\in{\Z_+}$, we
obtain that for $n\in{\Z_+}$
\begin{align*}
d_{\demi}\left( \Phi_0\circ\cdots\circ \Phi_{n}(\zeta_{n}) , \alpha_0 \right)
& =
d_{\demi}\left( \Phi_0\circ\cdots\circ \Phi_{n}(\zeta_{n}) ,
  \Phi_0\circ\cdots\circ \Phi_{n}(\alpha_{n+1}) \right) \\ 
& \leq \delta^{n}d_{\demi}\left( \Phi_{n}(\zeta_{n}) , \Phi_{n}(\alpha_{n+1}) \right)
\leq \eta\delta^{n}.
\end{align*}
Therefore, $\dhlim_{n\to\infty}\Phi_0\circ\cdots\circ
\Phi_n(\zeta_n)=\alpha_0$. 
\end{proof}

From now on, set 
\[\nu:= \nu_1\cdot \nu_2.\]
Now unlike in \cite{FHS}[Lemma 4.5] or in \cite{FHS5}[Proposition 3.4]
we shall not rely directly on a fixed point of $\Phi_n$ but on one of 
$\Phi_{ n}\circ\cdots\circ\Phi_{n+\nu -1}$. The proof is unfortunately
more complicated but the improvement is real as we can treat 
potentials satisfying $V-\tau^{\nu_1} V \in \ell^1$. Recall that with
the approach of \cite{FHS}, one covers only the case $\nu=1$. We
localize in energy and introduce:  
\begin{align}\label{e:K}
K_{x_1,x_2,\e}:=(x_1,x_2)+\rmi (0,\e) .
\end{align}

\begin{proposition}\label{p:born}
Take $x\in\left( -\sqrt{m^2+4} , -m\right)\cup\left( m ,
  \sqrt{m^2+4}\right)$, $\nu\in\Z_+\setminus\{0\}$, and assume that \eqref{e:spn1} and \eqref{e:spn2}
hold true, then there exist $x_1,x_2\in\R$ such that
$x\in (x_1 ,x_2 )$ and $M_1,\e >0$  so that
\[
d_\demi \left( \alpha_0 ,\rmi\right)
=d_\demi \left( \left\langle\delta_0 , \left(\Delta^{(0)}_{\lambda ,V,W}\right)^{-1}\delta_0\right\rangle ,\rmi\right)
 \leq M_1 
\]
for all $\lambda\in K_{x_1 ,x_2 ,\e}$. In particular there exists $M_2 >0$ such that
\[
\left| \left\langle\delta_0 , \left(\Delta^{({\Z_+})}_{2,\lambda ,V,W}\right)^{-1}\delta_0\right\rangle \right|
 \leq M_2 
\]
for all $\lambda\in K_{x_1 ,x_2 ,\e}$.
\end{proposition}
\begin{proof}
Using Lemma \ref{l:pol}, there exist some  polynomials 
$A,B_1,B_2,C\in \R[X_1, \ldots, X_{3\nu}]$ such that 
\begin{align*}
\Phi_{ n}\circ\cdots\circ\Phi_{n+\nu -1}(z) &=\varphi_{a_n, b_n,
  c_n}\circ\cdots\circ\varphi_{a_{n+\nu-1}, b_{n+\nu-1},
  c_{n+\nu-1}}(z)
\\&=
-\frac{C(\omega_{n, \lambda}
  )+B_2(\omega_{n, \lambda} )z}{B_1(\omega_{n, \lambda} )+A(\omega_{n, \lambda} )z}, 
\end{align*}
for all $\lambda\in\demi$ and $n\in{\Z_+}$, where
\begin{align*}
\omega_{n, \lambda}:= (a_{ n} ,b_{ n},c_{ n},\ldots  , a_{n+\nu -1}
,b_{n+\nu -1},c_{n+\nu -1}) , 
\end{align*}
and where $B_1(\omega )+A(\omega )z\not =0$ for all  $\omega\in({\rm
  cl}(\demi)^2\times \R^*_+)^{\nu}$ and $z\in\demi$. 

We now work in a neighbourhood of $x$. First notice that the fixed
points of $\varphi_{x+m ,x-m, 1}$ are given by:
\begin{align}\label{e:fp}
-\frac{x-m}{2}\pm\frac{1}{2}\rmi\sqrt{\frac{x-m}{x+m}(4+m^2-x^2)}. 
\end{align}
Then 
\[
 \underbrace{\varphi_{x+m
  ,x-m, 1}\circ\cdots \circ\varphi_{x+m
  ,x-m,1}}_{\nu \, {\rm times}}  = -\frac{C(\omega_{\infty ,x})+B_2(\omega_{\infty
    ,x})\cdot}{B_1(\omega_{\infty ,x})+A(\omega_{\infty
    ,x})\cdot},
\]
where $\omega_{\infty ,x}:=(x+m,x-m,1,\ldots ,x+m,x-m,1)$, has at
least \eqref{e:fp} as fixed points. As it is a homography it has exactly
at most two fixed points. Note also that $A (\omega_{\infty ,x})\not
=0$, because  there are two different fixed points. 

Now we would like to study the fixed points of 
\[
R(\omega, z):=-\frac{C(\omega )+B_2(\omega )z}{B_1(\omega )+A(\omega )z}
\]
with respect to $z$, for $\omega$ being in a neighbourhood of
$\omega_{\infty ,x}$. As the Inverse Function Theorem does not seem to
apply we rely on a direct approach.  
Since $A$ is continuous, there exists a
neighbourhood $\Omega_1$ of $\omega_{\infty ,x}$ such that
$A(\omega)\not=0$ for all $\omega\in\Omega_1$. We define on $\Omega_1$ 
\begin{align}\label{e:Z}
Z(\omega ):=-\frac{B_1(\omega )+B_2(\omega )}{2A(\omega
  )}+\frac{1}{2}\rmi\sqrt{4\frac{C(\omega )}{A(\omega
    )}-\left(\frac{B_1(\omega )+B_2(\omega )}{A(\omega )}\right)^2}, 
\end{align}
where we have chosen the square root in order to guarantee that:
\[
\Re\left(\sqrt{4\frac{C(\omega )}{A(\omega )}-\left(\frac{B_1(\omega )+B_2(\omega )}{A(\omega )}\right)^2}\right)\geq 0,
\]
for all $\omega\in\Omega_1$. A direct computation gives that $Z(\omega
)$ is a fixed point of $R(\omega, \cdot)$ on $\Omega_1$.  Since $A,
B_1,$ and $B_2$ are polynomials with real 
coefficients, we infer that  $\Im (Z(\omega_{\infty ,x}))\geq 0$, by
the choice of the square root. On the other hand $Z(\omega_{\infty
  ,x})$ belongs to \eqref{e:fp}. Therefore we infer that
\begin{align}\label{e:ptfx}
Z(\omega_{\infty
  ,x})=-\frac{x-m}{2}+\frac{1}{2}\rmi\sqrt{\frac{x-m}{x+m}(4+m^2-x^2)}\in\demi. 
\end{align}
In particular, since $A, B_1, B_2,$ and $C$ are polynomials with real
coefficients,
\[4\frac{C(\omega_{\infty ,x} )}{A(\omega_{\infty ,x}
  )}-\left(\frac{B_1(\omega_{\infty ,x} )+B_2(\omega_{\infty ,x} 
    )}{A(\omega_{\infty ,x} )}\right)^2=\frac{x-m}{x+m}(4+m^2-x^2)>0.\]
Therefore there exists a neighbourhood $\Omega_2\subset \Omega_1$ of
$\omega_{\infty ,x}$ such that 
\[
4\frac{C(\omega )}{A(\omega )}-\left(\frac{B_1(\omega )+B_2(\omega
    )}{A(\omega )}\right)^2\not\in\R_-, \mbox{ for all } \omega\in\Omega_2.
\]
We infer that we can take the principal value of the square root in
the definition of \eqref{e:Z} when $\omega \in \Omega_2$. In
particular, $Z\in \Cc^\infty(\Omega_2 ,\C )$. Hence, recalling
\eqref{e:ptfx}, there exists a compact neighbourhood $\Omega_3\subset
\Omega_2$ of $\omega_{\infty ,x}$ and $\eta_1,M_1>0$ such that 
\[\Im (Z(\omega))>\eta_1  \mbox{ and } |Z(\omega)|\leq M_1,\] 
for all
$\omega\in\Omega_3$. Now there exist $x_1,x_2\in\R$, $\e>0$, and
$n_0\in{\Z_+}$ such that $x\in (x_1 
,x_2)$ and  $\omega_{n, \lambda} \in\Omega_3$, for all $\lambda\in
K_{x_1,x_2,\e}$ and $n\geq n_0$. We define now
\begin{align*}
z_n(\lambda ):= 
Z(\omega_{n ,\lambda} )  
\end{align*}
for all $\lambda\in K_{x_1,x_2,\e}$ and $n\geq n_0$. Notice that 
\begin{align}\label{e:inz}
\Im
(z_n(\lambda ))>\eta_1 \mbox{ and } |z_n(\lambda )|\leq M_1,
\end{align} 
for all
$\lambda\in K_{x_1,x_2,\e}$ and $n\geq n_0$. Moreover, by definition
of $Z$ we have 
\begin{align}\label{e:infp}
\Phi_{ n}\circ\cdots\circ\Phi_{n+\nu -1}(z_n(\lambda ))=z_n(\lambda )
\end{align}
for all $\lambda\in K_{x_1,x_2,\e}$ and $n\geq n_0$. 
Next, there is
$M_2>0$ such that 
\begin{align*}
\left( |a_{k}-a_{k+\nu}| +|b_{k}-b_{k+\nu}| + |c_{k}-c_{k+\nu}|\right)
\leq
M_2\left( \| V(k)-V(k+\nu )\| +\| W(k)-W(k+\nu )\|\right)
\end{align*}
for all $\lambda\in K_{x_1,x_2,\e}$ and $k\in{\Z_+}$. Now since $Z$ is  
$\Cc^\infty(\Omega_3 ,\C )$ and $\Omega_3$ is compact, there exists a
Lipschitz constant $M_3>0$ such that 
\begin{align}
\notag
|z_{n+\nu}(\lambda )-z_n(\lambda )|
\leq & \,M_3\| \omega_{n+\nu, \lambda} - \omega_{n, \lambda}\|
\\
\nonumber
\leq & \,
M_3 \left( |a_{n+\nu}(\lambda)-a_{n}(\lambda)|
  +|b_{n+\nu}(\lambda)-b_{n}(\lambda)| +
  |c_{n+\nu}(\lambda)-c_{n}(\lambda)| 
 + \cdots \right. \\ 
\nonumber & \hspace*{-1cm} \left.   +
 |a_{n+2\nu -1}(\lambda)-a_{n+\nu -1}(\lambda)| +|b_{n+2\nu
   -1}(\lambda)-b_{n+\nu -1}(\lambda) | + |c_{n+2\nu
   -1}(\lambda)-c_{n+\nu -1}(\lambda)|\right) 
\\
\label{e:in1}
\leq & \, M_2 M_3 \sum_{k=0}^{\nu-1} \left (\|V(n+k+\nu)- V(n+k)\| +
  \|W(n+k+\nu)- W(n+k)\| \right)
\end{align}
for all $\lambda\in K_{x_1,x_2,\e}$ and $n\geq n_0$. 
By Corollary \ref{c:cvg}, for all $\lambda\in K_{x_1,x_2 , \e}$, we have:
\begin{align}
\notag d_{\demi}(\rmi ,\alpha_0) 
= & \lim_{n\to\infty}d_{\demi}(\rmi ,
\Phi_0\circ\cdots\circ\Phi_{n_0+\nu (n+1)-1}(z_{n_0+\nu n}(\lambda ))
) \\ \notag
\leq &\lim_{n\to\infty}\left( d_{\demi}(\rmi , \Phi_0(\rmi) )
  +\sum_{k=0}^{n_0+\nu
    -3}d_{\demi}(\Phi_0\circ\cdots\circ\Phi_{k}(\rmi )
  ,\Phi_0\circ\cdots\circ\Phi_{k+1}(\rmi ) ) 
\right. \\ \notag & \quad \quad \quad \quad 
+ d_{\demi}(\Phi_0\circ\cdots\circ\Phi_{n_0+\nu -2}(\rmi )
,\Phi_0\circ\cdots\circ\Phi_{n_0+\nu -1}(z_{n_0}) ) 
\\ \notag & \quad \quad \quad \quad  \left.
+ \sum_{k=0}^{n-1}d_{\demi}(\Phi_0\circ\cdots\circ\Phi_{n_0+\nu
  (k+1)-1}(z_{n_0+\nu k} ) ,\Phi_0\circ\cdots\circ\Phi_{n_0+\nu
  (k+2)-1}(z_{n_0+\nu (k+1)} ) )\right) \\ 
\label{e:in10}  \leq  & 
\sum_{k= 0}^{n_0+\nu -2} d_{\demi}(\rmi,\Phi_{k}(\rmi ))+
d_{\demi}(\rmi,\Phi_{{n_0+\nu -1}}(z_{n_0}))  
\\ \notag & \quad \quad \quad \quad
+ \sum_{k\geq 0}d_{\demi}(z_{n_0+\nu k}  ,\Phi_{n_0+\nu
  (k+1)}\circ\cdots\circ\Phi_{n_0+\nu (k+2)-1}(z_{n_0+\nu (k+1)} ) )\\ 
 \label{e:in09}  =  &
\sum_{k= 0}^{n_0+\nu -2} d_{\demi}(\rmi,\Phi_{k}(\rmi ))+
d_{\demi}(\rmi,\Phi_{{n_0+\nu -1}}(z_{n_0})) + \sum_{k\geq
  0}d_{\demi}(z_{n_0+\nu k}  ,z_{n_0+\nu (k+1)}  )\\ 
  \leq  &  \label{e:in0}
\sum_{k= 0}^{n_0+\nu -2} d_{\demi}(\rmi,\Phi_{k}(\rmi ))+
d_{\demi}(\rmi,\Phi_{{n_0+\nu -1}}(z_{n_0})) + \sum_{k\geq
  0}\frac{\left| z_{n_0+\nu k}  - z_{n_0+\nu (k+1)} \right|}{\left(\Im
    (z_{n_0+\nu k} )\right)^{1/2}\left( \Im  ( z_{n_0+\nu (k+1)})
  \right)^{1/2}}  . 
\end{align}
Here in \eqref{e:in10} we have used the fact that $\Phi_n$ is a
contraction, in \eqref{e:in09} we exploited \eqref{e:infp}, and 
in \eqref{e:in0} we relied on \eqref{inh}. 

Coming back to \eqref{e:bn}, one finds easily $M_4,\eta_4 > 0$ such that
\begin{align*}
\max\left(|a_n|, |b_n| \right)\leq M_4 \quad\text{ and }\quad \eta_4 <
c_n\leq M_4, 
\end{align*}
for all $\lambda\in K_{x_1,x_2 ,\e}$ and for all
$n\in{\Z_+}$. Then, Lemma \ref{l:phiim} ensures that
\begin{align}\label{e:in3}
d_{\demi}(\Phi_n(z),\rmi)\leq \left(\frac{(M_4+M_4|z|)^2}{\eta_4\Im
    (z)}+1\right)\frac{(M_4|z|+M_4)\left( M_4+(\eta_4\Im (z
    ))^{-1}\right)}{\sqrt{\eta_4\Im (z)}}.
\end{align}
for all $z\in\demi$, $\lambda\in K_{x_1 ,x_2 ,\e}$, and $n\in {\Z_+}$. 
Finally combining \eqref{e:in0} and estimates \eqref{e:inz},
\eqref{e:in1}, and \eqref{e:in3}, we infer:
\begin{align*}
d_{\demi}(\alpha_0 ,\rmi)
\leq & 
(n_0+\nu-1)  \left(\frac{(2M_4)^2}{\eta_4}+1\right)\frac{2M_4\left(
    M_4+\eta_4^{-1}\right)}{\sqrt{\eta_4}} 
\\ & +
\left(\frac{(M_4+M_4M_1)^2}{\eta_4\eta_1}+1\right)\frac{(M_4M_1+M_4)\left(
    M_4+(\eta_4\eta_1)^{-1}\right)}{\sqrt{\eta_4\eta_1}} 
\\ & +
 \frac{M_2M_3}{\eta_1} \left( \|V-\tau^\nu V\|_1 +  \|W-\tau^\nu W\|_1 \right) ,
\end{align*}
for all $\lambda\in K_{x_1,x_2,\e}$.
The second point comes by recalling that $\Delta^{(0)}_{\lambda ,V,W}=\Delta^{({\Z_+}
  )}_{2,\lambda ,V,W}$. 
\end{proof}

\section{The absolutely continuous spectrum}\label{s:abso}
We recall the following standard result, e.g., \cite{ReSi}[Theorem
XIII.19].
\begin{theorem}\label{t:inac}
Let $H$ be a self-adjoint operator of $\Hc$, let $(x_1,x_2)$ be an
interval and $f\in \Hc$. Suppose 
\begin{align}\label{e:inac1}
\limsup_{\e \downarrow 0^+}\sup_{x\in (x_1,x_2)}\left|\left\langle f
    ,(H-(x+\rmi\e ))^{-1}f\right\rangle\right| < +\infty , 
\end{align}
then the measure $\langle f, \bone_{(\cdot) }(H)f \rangle$ is purely
absolutely continuous w.r.t.\ the Lebesgue measure on $(x_1, x_2)$. 
\end{theorem}

\subsection{The case of ${\Z_+}$}\label{s:N}
In the previous section we have estimated the resolvent of
$\Delta^{({\Z_+})}_{2,\lambda ,V,W}$.  Keeping in mind  Proposition
\ref{p:schur}, we explain how to transfer the result to $H^{({\Z_+}
  )}_{m,V,W}$. We start by reducing the study to a unique vector.

\begin{lemma}\label{l:maj}
Given $A\subset \C\backslash\R$ bounded, $V\in \ell^\infty ({\Z_+},\R^2)$,
and $W\in \ell^\infty ({\Z_+},\C^2)$ with $W_1(n)\not =-1$ and $W_2(n)\not
=1$ for all $n\in{\Z_+}$. Suppose that there exists $C_1>0$ such that 
\[
\left| \left\langle \left(\begin{array}{c}
0\\
1
\end{array}\right)\delta_{0}
,
\left( H^{({\Z_+} )}_{m,V,W}-\lambda\right)^{-1}
\left(\begin{array}{c}
0\\
1
\end{array}\right)\delta_{0} \right\rangle \right|
\leq C_1 ,
\]
for all $\lambda\in A$. Then for all $x_1,y_1,x_2,y_2\in\C$
and $n_1,n_2\in{\Z_+}$ there exists $C_2>0$ such that 
\[
\left|\left\langle \left(\begin{array}{c}
x_1\\
y_1
\end{array}\right)\delta_{n_1}
,
\left( H^{({\Z_+} )}_{m,V,W}-\lambda\right)^{-1}
\left(\begin{array}{c}
x_2\\
y_2
\end{array}\right)\delta_{n_2} \right\rangle \right|
\leq C_2 ,
\]
for all $\lambda\in A$.
\end{lemma}
\begin{proof}
Let
\[
f:=\left( H^{({\Z_+} )}_{m,V,W}-\overline{\lambda}\right)^{-1}
\left(\begin{array}{c}
0\\
1
\end{array}\right)\delta_{0} .
\]
We have clearly
\[
|f_2(0)|=\left| \left\langle \left(\begin{array}{c}
0\\
1
\end{array}\right)\delta_{0}
,
\left( H^{({\Z_+} )}_{m,V,W}-\lambda\right)^{-1}
\left(\begin{array}{c}
0\\
1
\end{array}\right)\delta_{0} \right\rangle \right|
\leq C_1 ,
\]
for all $\lambda\in A$. By definition, $f$ is the unique solution in
$\ell^2({\Z_+},\C^2)$ of 
\begin{align*}
\left\lbrace\begin{array}{l}
(m+V_1(n)-\overline{\lambda} )f_1(n)+(1+W_1(n))f_2(n)+(-1+W_2(n))f_2(n+1) =0\\
(1+\overline{W_1}(n))f_1(n)+(-1+\overline{W_2}(n-1))f_1(n-1)+(-m+V_2(n)-\overline{\lambda} )f_2(n)=0
\end{array}\right.
\end{align*}
for all $n \geq 1$ and of
\begin{align*}
\left\lbrace\begin{array}{l}
(m+V_1(0)-\overline{\lambda} )f_1(0)+(1+W_1(0))f_2(0)+(-1+W_2(0))f_2(1) =0\\
(1+\overline{W_1}(0))f_1(0)+(-m+V_2(0)-\overline{\lambda} )f_2(0)=1 .
\end{array}\right.
\end{align*}
So by induction on $n\in{\Z_+}$ there exists $D_n>0$ such that
\[
| f_i(n) |\leq D_n  ,
\]
for all $\lambda\in A$ and $i\in \{ 1, 2\}$. Therefore we obtain that
\begin{align*}
\left| \left\langle \left(\begin{array}{c}
0\\
1
\end{array}\right)\delta_{0}
,
\left( H^{({\Z_+} )}_{m,V,W}-\lambda\right)^{-1}
\left(\begin{array}{c}
x_2\\
y_2
\end{array}\right)\delta_{n_2} \right\rangle \right|
& =\left| \left\langle \left(\begin{array}{c}
x_2\\
y_2
\end{array}\right)\delta_{n_2}
,
\left( H^{({\Z_+} )}_{m,V,W}-\overline{\lambda}\right)^{-1}
\left(\begin{array}{c}
0\\
1
\end{array}\right)\delta_{0} \right\rangle \right| \\
& =|x_2f_1(n_2)+y_2f_2(n_2)| \leq (|x_2|+|y_2|)D_{n_2}=:C_2 .
\end{align*}
for all $\lambda\in A$. Now let
\[
g:=\left( H^{({\Z_+} )}_{m,V,W}-\lambda\right)^{-1}
\left(\begin{array}{c}
x_2\\
y_2
\end{array}\right)\delta_{n_2} ,
\]
we have
\[
|g_2(0)|=\left| \left\langle \left(\begin{array}{c}
0\\
1
\end{array}\right)\delta_{0}
,
\left( H^{({\Z_+} )}_{m,V,W}-\lambda\right)^{-1}
\left(\begin{array}{c}
x_2\\
y_2
\end{array}\right)\delta_{n_2} \right\rangle \right| \leq C_2 .
\]
for all $\lambda\in A$. By definition, $g$ is the unique solution in $\ell^2({\Z_+},\C^2)$ of
\begin{align*}
\left\lbrace\begin{array}{l}
(m+V_1(n)-\lambda )g_1(n)+(1+W_1(n))g_2(n)+(-1+W_2(n))g_2(n+1) =x_2\delta_{n_2}(n)\\
(1+\overline{W_1}(n))g_1(n)+(-1+\overline{W_2}(n-1))g_1(n-1)+(-m+V_2(n)-\lambda
)g_2(n)=y_2\delta_{n_2}(n) 
\end{array}\right.
\end{align*}
for all $n \geq 1$ and of
\begin{align*}
\left\lbrace\begin{array}{l}
(m+V_1(0)-\lambda )g_1(0)+(1+W_1(0))g_2(0)+(-1+W_2(0))g_2(1) =x_2\delta_{n_2}(0)\\
(1+\overline{W_1}(0))g_1(0)+(-m+V_2(0)-\lambda )g_2(0)=y_2\delta_{n_2}(0) .
\end{array}\right.
\end{align*}
So by induction for all $n\in{\Z_+}$ there exists $D_n^\prime >0$ such that
\[
| g_i(n) |\leq D_n^\prime ,
\]
for all $\lambda\in A$ and $i\in \{ 1, 2\}$. Therefore we obtain that
\begin{align*}
\left| \left\langle \left(\begin{array}{c}
x_1\\
y_1
\end{array}\right)\delta_{n_1}
,
\left( H^{({\Z_+} )}_{m,V,W} -\lambda\right)^{-1}
\left(\begin{array}{c}
x_2\\
y_2
\end{array}\right)\delta_{n_2} \right\rangle \right|
=|x_1g_1(n_1)+y_1g_2(n_1)| \leq (|x_1|+|y_1|)D_{n_1}^\prime=:C_3 ,
\end{align*}
for all $\lambda\in A$. This concludes the proof.
\end{proof}

We are now in position to conclude with our main result.

\begin{proof}[Proof of Theorem \ref{t:spn}]
Since $D_m^{({\Z_+})}-H^{({\Z_+} )}_{m,V,W}$ is compact, then Weyl Theorem
gives the first point. We prove the second one. 
Take $x\in\left(-\sqrt{m^2+4} , -m\right)\cup \left(m , \sqrt{m^2+4}\right)$.
By Propositions \ref{p:schur} and \ref{p:born}  we have that there
exists $\e , C>0$ and $x_1,x_2\in\R$ such that $x\in(x_1,x_2)$ and  
\begin{align*}
& \left| \left\langle \left(\begin{array}{c}
0\\
1
\end{array}\right)\delta_{0}
,
\left( H^{({\Z_+} )}_{m,V,W}-\lambda \right)^{-1}
\left(\begin{array}{c}
0\\
1
\end{array}\right)\delta_{0} \right\rangle \right| 
=   \left| \left\langle \delta_{0}
,
(\Delta^{({\Z_+} )}_{2,\lambda  ,V,W})^{-1}
\delta_{0} \right\rangle \right| \leq C 
\end{align*}
for all $\lambda\in K_{x_1,x_2,\e}$. Then Theorem \ref{t:inac} and
Lemma \ref{l:maj} conclude by density. 
\end{proof}

\subsection{The case of $\Z$}\label{s:Z}
Now we  express $\Delta^{(\Z )}_{2,\lambda ,V,W}$ with the
help of $\Delta^{(0)}_{\lambda ,\cdot ,\cdot }$. 

\begin{lemma}\label{l:rec}
Take $\lambda\in\demi$,  $V\in \ell^\infty (\Z,\R^2)$,  $W\in
\ell^\infty (\Z,\C^2)$ with $W_1(0)\not =-1$ and $W_2(-1)\not =1$. Set:
\begin{align*}
\Phi(z_1,z_2) & := - \left(
a-\left(
b +cz_1
\right)^{-1}
-\left(
b^\prime +c^\prime z_2
\right)^{-1}
\right)^{-1} \\
a & := \lambda +m -V_2(0)\in\demi , \\
b & :=(\lambda -m -V_1(0))|1+W_1(0)|^{-2}\in\demi , \quad b^\prime  :=(\lambda -m -V_1(-1))|1-W_2(-1)|^{-2}\in\demi ,\\
c & := \left|\frac{1-W_2(0)}{1+W_1(0)}\right|^2\in\R_+^* ,
\quad\quad\quad\quad\quad\quad\quad \quad \!\!\!
c^\prime  := \left|\frac{1+W_1(-1)}{1-W_2(-1)}\right|^2\in\R_+^* ,
\end{align*}
we have
\begin{align*}
& \left\langle \delta_0 ,\left(\Delta^{(\Z )}_{2,\lambda ,V,W}\right)^{-1}\delta_0 \right\rangle \\
& \quad\quad\quad\quad\quad
=\Phi\left(
\left\langle \delta_0 ,\Delta^{(0)}_{\lambda ,\left(\tau V\right)_{|_{\Z_+}},\left(\tau W\right)_{|_{\Z_+}}}\delta_0 \right\rangle
,
\left\langle \delta_0 ,\Delta^{(0)}_{\lambda ,\left(\tau^2 SV_1
      ,\tau S V_2 \right)^t_{|_{\Z_+}},\left(\tau -\tau^2 S W_2 , -\tau^2
      S W_1 \right)^t_{|_{\Z_+}}}\delta_0 \right\rangle
\right) .
\end{align*}
\end{lemma}
\begin{proof}
We define
\begin{align*}
&
f:=\left(\Delta^{(\Z )}_{2,\lambda ,V,W}\right)^{-1}\delta_0, \\
&
g:=\left(\Delta^{(0)}_{\lambda ,\left(\tau
      V\right)_{|_{\Z_+}},\left(\tau
      W\right)_{|_{\Z_+}}}\right)^{-1}\delta_{0}
\mbox{ and }
h:=\left(\Delta^{(0)}_{\lambda ,\left(\tau^2 S V_1 ,\tau  SV_2
    \right)^t_{|_{\Z_+}},\left(\tau -\tau^2 S W_2 , -\tau^2 SW_1 \right)^t_{|_{\Z_+}}}\right)^{-1}\delta_{0} .
\end{align*}
Clearly
\begin{align*}
& \left\langle \delta_0 ,\left(\Delta^{(\Z )}_{2,\lambda
      ,V,W}\right)^{-1}\delta_0 \right\rangle=f(0) , \quad
\left\langle \delta_0 ,\Delta^{(0)}_{\lambda ,\left(\tau
      V\right)_{|_{\Z_+}},\left(\tau W\right)_{|_{\Z_+}}}\delta_0
\right\rangle=g(0) \\ 
& \left\langle \delta_0 ,\Delta^{(0)}_{\lambda ,\left(\tau^2 S V_1
      ,\tau S V_2 \right)^t_{|_{\Z_+}},\left(\tau -\tau^2 S W_2 , -\tau^2
      S W_1 \right)^t_{|_{\Z_+}}}\delta_0 \right\rangle=h(0)  .
\end{align*}
By definition $f$ is the unique solution in $\ell^2(\Z,\C)$ of
$\Delta^{(\Z )}_{2,\lambda ,V,W}f=\delta_0$, i.e., 
\begin{align}
\notag &
\frac{1+\overline{W_1}(n)}{\lambda -m -V_1(n)}\left((1+W_1(n))f(n)+(-1+W_2(n))f(n+1)\right) \\
\label{rec2} &  \quad \quad \quad \quad \quad 
+\frac{1-\overline{W_2}(n-1)}{\lambda -m -V_1(n-1)}\left((1-W_2(n-1))f(n)+(-1-W_1(n-1))f(n-1)\right) \\
\notag & \quad \quad \quad \quad \quad \quad \quad \quad \quad \quad 
 -(\lambda +m-V_2(n))f(n)  =\delta_0(n),
\end{align}
for all $n\in\Z$.
Let $f^\prime:=\left(\tau f\right)_{|_{\Z_+}}$, we see that $f^\prime$ is solution of
\begin{align*}
& \frac{1+\overline{W_1}(n+1)}{\lambda -m -V_1(n+1)}\left((1+W_1(n+1))f^\prime(n)+(-1+W_2(n+1))f^\prime(n+1)\right) \\
\notag &  \quad \quad \quad \quad \quad 
+\frac{1-\overline{W_2}(n)}{\lambda -m -V_1(n)}\left((1-W_2(n))f^\prime(n)+(-1-W_1(n))f^\prime(n-1)\right) \\
\notag & \quad \quad \quad \quad \quad  \quad \quad \quad \quad \quad 
 -(\lambda +m-V_2(n+1))f^\prime(n)  =0
\end{align*}
for all $n\geq 1$ and
\begin{align*}
& \frac{1+\overline{W_1}(1)}{\lambda -m -V_1(1)}\left((1+W_1(1))f^\prime(0)+(-1+W_2(1))f^\prime(1)\right) 
 -(\lambda +m-V_2(1))f^\prime(0)\\
&  \quad \quad \quad \quad \quad =
  \frac{1-\overline{W_2}(0)}{\lambda -m -V_1(0)}\left((1+W_1(0))f(0)+(-1+W_2(0))f(1)\right) .
\end{align*}
So we obtain that
\begin{align*}
\Delta^{(0)}_{\lambda ,\left(\tau V\right)_{|_{\Z_+}},\left(\tau W\right)_{|_{\Z_+}}}f^\prime
 & =\frac{1-\overline{W_2}(0)}{\lambda -m -V_1(0)}\left((1+W_1(0))f(0)+(-1+W_2(0))f(1)\right)\delta_{0} \\
 & =
 \frac{1-\overline{W_2}(0)}{\lambda -m -V_1(0)}\left((1+W_1(0))f(0)+(-1+W_2(0))f(1)\right)\Delta^{(0)}_{\lambda ,\left(\tau V\right)_{|_{\Z_+}},\left(\tau W\right)_{|_{\Z_+}}}g .
\end{align*}
Since $\Delta^{(0)}_{\lambda ,\left(\tau V\right)_{|_{\Z_+}},\left(\tau W\right)_{|_{\Z_+}}}$ is invertible, we get
\begin{align}\label{e:rec1}
\frac{1-\overline{W_2}(0)}{\lambda -m -V_1(0)}\left((1+W_1(0))f(0)+(-1+W_2(0))f(1)\right)g(0)
=f^\prime(0)=f(1) .
\end{align}
Now let $f^\pprime:=\left(\tau S f\right)_{|_{\Z_+}}$
We see that $f^\pprime$ is solution of
\begin{align*}
& 
\frac{1+\overline{W_1}(-n-1)}{\lambda -m -V_1(-n-1)}\left((1+W_1(-n-1))f^\pprime(n)+(-1+W_2(-n-1))f^\pprime(n-1)\right) \\
\notag &  \quad \quad \quad \quad \quad 
+\frac{1-\overline{W_2}(-n-2)}{\lambda -m -V_1(-n-2)}\left((1-W_2(-n-2))f^\pprime(n)+(-1-W_1(-n-2))f^\pprime(n+1)\right) \\
\notag & \quad \quad \quad \quad \quad  \quad \quad \quad \quad \quad 
 -(\lambda +m-V_2(-n-1))f^\pprime(n)  =0
\end{align*}
for all $n\geq 1$ and 
\begin{align*}
 & 
\frac{1-\overline{W_2}(-2)}{\lambda -m -V_1(-2)}\left((1-W_2(-2))f^\pprime(0)+(-1-W_1(-2))f^\pprime(1)\right) 
 -(\lambda +m-V_2(-1))f^\pprime(0) \\
\notag  &  \quad \quad \quad \quad \quad =
\frac{1+\overline{W_1}(-1)}{\lambda -m -V_1(-1)} \left((1-W_2(-1))f(0)+(-1-W_1(-1))f(-1)\right).
\end{align*}
By setting $\tilde{\Delta}:=\Delta^{(0)}_{\lambda ,\left(\tau^2 S
    V_1 ,\tau S V_2 \right)^t_{|_{\Z_+}},\left(\tau -\tau^2 S W_2 ,
    -\tau^2 S W_1 \right)^t_{|_{\Z_+}}}$, we obtain
\begin{align*}
\tilde{\Delta}f^\pprime
 & =\frac{1+\overline{W_1}(-1)}{\lambda -m -V_1(-1)} \left((1-W_2(-1))f(0)+(-1-W_1(-1))f(-1)\right)\delta_{0} \\
 & =
 \frac{1+\overline{W_1}(-1)}{\lambda -m -V_1(-1)} \left((1-W_2(-1))f(0)+(-1-W_1(-1))f(-1)\right)\tilde{\Delta}h .
\end{align*}
Since $\tilde{\Delta}$ is invertible, we infer
\begin{align}\label{e:rec2}
\frac{1+\overline{W_1}(-1)}{\lambda -m -V_1(-1)} \left((1-W_2(-1))f(0)+(-1-W_1(-1))f(-1)\right)h(0)
=f^\pprime(0)=f(-1) .
\end{align}
Straightforwardly, using \eqref{rec2} for $n=0$, \eqref{e:rec1} and \eqref{e:rec2} we have that $f(0) = \Phi(g(0),h(0))$.
\end{proof}

Now with this Lemma we can obtain that $\left\langle \delta_0
  ,\left(\Delta^{(\Z )}_{2,\lambda ,V,W}\right)^{-1}\delta_0
\right\rangle $ is bounded independently of $\lambda$ with the
Proposition \ref{p:born}.

\begin{corollary}\label{c:born}
Take $x\in\left(-\sqrt{m^2+4} , -m\right)\cup \left(m ,
  \sqrt{m^2+4}\right)$, $V\in \ell^\infty (\Z,\R^2)$, and $W\in
\ell^\infty (\Z,\C^2)$ with $W_1(n)\not =-1$ and $W_2(n)\not=1$, for
all $n\in\Z$. Suppose that there exists $\nu\in\Z_+\setminus\{0\}$ such that
\eqref{e:spz2} or \eqref{e:spz2'} holds true. 
Then there exist $C,\e >0$ and $x_1,x_2\in\R$ such that $x\in (x_1,x_2)$ and
\[
\left|
\left\langle \delta_0 ,\left(\Delta^{(\Z )}_{2,\lambda ,V,W}\right)^{-1}\delta_0 \right\rangle
\right|
 \leq C 
\]
for all $\lambda\in K_{x_1 ,x_2 ,\e}$.
\end{corollary}
\begin{proof}
Let
\begin{align*}
& \alpha_\lambda:= \left\langle \delta_0 ,\Delta^{(0)}_{\lambda ,\left(\tau V\right)_{|_{\Z_+}},\left(\tau W\right)_{|_{\Z_+}}}\delta_0 \right\rangle \in\demi , \\
& \alpha_\lambda^\prime:=\left\langle \delta_0
  ,\Delta^{(0)}_{\lambda ,\left(\tau^2 S V_1 ,\tau  SV_2
    \right)^t_{|_{\Z_+}},\left(\tau -\tau^2 SW_2 , -\tau^2 SW_1 \right)^t_{|_{\Z_+}}}\delta_0 \right\rangle \in\demi .
\end{align*}
With Lemma \ref{l:rec} we have
\begin{align*}
\left|
\left\langle \delta_0 ,\left(\Delta^{(\Z )}_{2,\lambda ,V,W}\right)^{-1}\delta_0 \right\rangle
\right|
& =
\left|
a-\left(b+c\alpha_\lambda\right)^{-1}-\left(b^\prime + c^\prime\alpha^\prime_\lambda\right)^{-1}
\right|^{-1} \\
& \leq
\min \left( \left(\Im \left( -\left(b+ c\alpha_\lambda\right)^{-1}\right)\right)^{-1}
,
\left(\Im \left( -\left(b^\prime + c^\prime\alpha^\prime_\lambda\right)^{-1}\right)\right)^{-1}
\right) .
\end{align*}
Now if $V_{|_{{\Z_+}}}-\tau^{\nu_1} V_{|_{{\Z_+}}},W_{|_{{\Z_+}}}-\tau^{\nu_2} W_{|_{{\Z_+}}} \in\ell^1 ({\Z_+},\C^2)$ with Proposition \ref{p:born} there exist $C_1,\e >0$ and $x_1,x_2\in\R$ such that $x\in (x_1,x_2)$ and $d_\demi ( \alpha_\lambda ,\rmi )\leq C_1$ for all $\lambda\in K_{x_1 ,x_2 ,\e}$, so there exists $C_2 >0$ such that
\[
\left|
\left\langle \delta_0 ,\left(\Delta^{(\Z )}_{2,\lambda ,V,W}\right)^{-1}\delta_0 \right\rangle
\right| \leq
\left( \Im \left( -\left( b+ c\alpha_\lambda\right)^{-1}\right)\right)^{-1}\leq C_2
\]
for all $\lambda\in K_{x_1 ,x_2 ,\e}$.
Now if
$V_{|_{\Z_-}}-\tau^{\nu_1} V_{|_{\Z_-}},W_{|_{\Z_-}}-\tau^{\nu_2} W_{|_{\Z_-}} \in\ell^1 (\Z_-,\C^2)$ with Proposition \ref{p:born} there exist $C_1,\e >0$  and
$x_1,x_2\in\R$ such that $x\in (x_1,x_2)$ and $d_\demi (
\alpha^\prime_\lambda ,\rmi )\leq C_1$ for all $\lambda\in K_{x_1 ,x_2
  ,\e}$, so there exists $C_2 >0$ such that 
\[
\left|
\left\langle \delta_0 ,\left(\Delta^{(\Z )}_{2,\lambda ,V,W}\right)^{-1}\delta_0 \right\rangle
\right| \leq
\left( \Im \left( -\left( b^\prime+ c^\prime\alpha^\prime_\lambda\right)^{-1}\right)\right)^{-1}\leq C_2
\]
for all $\lambda\in K_{x_1 ,x_2 ,\e}$.
\end{proof}

Finally we conclude with the help of the symmetry of charge. 

\begin{proof}[Proof of Theorem \ref{t:spz}]
Since $D_m^{(\Z)}-H^{(\Z )}_{m,V,W}$ is compact, then Weyl Theorem
gives the first point. We turn to the second one. 
Let $n\in\Z$, let $x\in\left(-\sqrt{m^2+4} , -m\right)\cup \left(m ,
  \sqrt{m^2+4}\right)$.  
Proposition \ref{p:schur} and Corollary
\ref{c:born} ensure that there exist $\e_1 ,C_1>0$ and $x_1,x_2\in\R$ such that
$x\in (x_1,x_2)$ and  
\begin{align*}
 \left| \left\langle \left(\begin{array}{c}
0\\
1
\end{array}\right)\delta_{n}
,
\left( H^{(\Z )}_{m,V,W}-\lambda \right)^{-1}
\left(\begin{array}{c}
0\\
1
\end{array}\right)\delta_{n} \right\rangle \right|& =
 \left| \left\langle \left(\begin{array}{c}
0\\
1
\end{array}\right)\delta_{0}
,
\left( H^{(\Z )}_{m,\tau^nV , \tau^nW}-\lambda \right)^{-1}
\left(\begin{array}{c}
0\\
1
\end{array}\right)\delta_{0} \right\rangle \right| \\
&=   \left| \left\langle \delta_{0}
,
(\Delta^{(\Z )}_{2,\lambda  ,\tau^nV , \tau^nW})^{-1}
\delta_{0} \right\rangle \right|\leq C_1,
\end{align*}
for all $\lambda\in K_{x_1,x_2,\e_1}$. Then, Theorem \ref{t:inac}
yields that the measure 
\[
\left\langle \left(\begin{array}{c}
0\\
1
\end{array}\right)\delta_{n}, \bone_{(\cdot) }(H^{({\Z_+} )}_{m,V,W}) \left(\begin{array}{c}
0\\
1
\end{array}\right)\delta_{n} \right\rangle
\]
is purely absolutely continuous on $(x_1, x_2)$. 
Now we use $U$, see \eqref{e:sym}. Let
\[
V^\prime:=\left( -S\tau^nV_2 , -S\tau^nV_1 \right)^t
\quad\text{ and }\quad
W^\prime:=\left( S\tau^n\overline{W_1} ,  S\tau^{n-1}\overline{W_2}\right)^t ,
\]
there exist $\e_2 ,C_2>0$ and $x_3,x_4\in\R$ such that $x\in (x_3,x_4)$ and
\begin{align*}
 \left| \left\langle \left(\begin{array}{c}
1\\
0
\end{array}\right)\delta_{n}
,
\left( H^{(\Z )}_{m,V,W}-\lambda \right)^{-1}
\left(\begin{array}{c}
1\\
0
\end{array}\right)\delta_{n} \right\rangle \right| &=
 \left| \left\langle \left(\begin{array}{c}
1\\
0
\end{array}\right)\delta_{0}
,
\left( H^{(\Z )}_{m,\tau^nV , \tau^nW}-\overline{\lambda} \right)^{-1}
\left(\begin{array}{c}
1\\
0
\end{array}\right)\delta_{0} \right\rangle \right| 
\\
&= 
\left| \left\langle \left(\begin{array}{c}
0\\
\delta_{0}
\end{array}\right)
,
\left( H^{(\Z )}_{m,V^\prime , W^\prime } -\left(-\overline{\lambda}\right) \right)^{-1}
\left(\begin{array}{c}
0\\
\delta_{0}
\end{array}\right) \right\rangle \right| 
\\
&=   \left| \left\langle \delta_{0}
,
(\Delta^{(\Z )}_{2,-\overline{\lambda} ,V^\prime , W^\prime })^{-1}
\delta_{0} \right\rangle \right|\leq C_2
\end{align*}
for all $\lambda\in K_{x_3,x_4,\e_2}$. Again Theorem \ref{t:inac}
gives that the measure 
\[
\left\langle \left(\begin{array}{c}
1\\
0
\end{array}\right)\delta_{n}, \bone_{(\cdot) }(H^{({\Z_+} )}_{m,V,W}) \left(\begin{array}{c}
1\\
0
\end{array}\right)\delta_{n} \right\rangle
\]
is purely absolutely continuous  $(x_3, x_4)$.
Finally, remembering that $x$ is arbitrary and by an argument of density, we
infer that $H^{({\Z_+} )}_{m,V,W}$ has pure ac spectrum
$\left(-\sqrt{m^2+4} , -m\right)\cup \left(m , \sqrt{m^2+4}\right)$.  
\end{proof}

\section{The case of the Laplacian}\label{s:Lap}
In this section we explain briefly how to adapt our proofs in order to
prove Theorem \ref{t:main0}. For the sake of clarity, we stick to the
case $\Delta +V $ and compare our proof directly to \cite{FHS}. 

We start by the case of ${\Z_+}$. 
\begin{theorem}\label{t:main0N}
Take $V\in \ell^\infty ({\Z_+},\R)$ and $\tau\in {\Z_+}$ such that:
\begin{align}\label{e:main0N}
\left.\begin{array}{c}
\lim_{n\to + \infty}V(n)=0,
\\
\\
V-\tau^\nu V\in\ell^1 ({\Z_+} ,\R), 
\end{array}\right. 
\end{align}
then the spectrum of $\Delta + V $ is purely absolutely continuous on
$(0, 4)$. 
\end{theorem}

Apart from Proposition \ref{p:change},
our presentation is very close to the one of
\cite{FHS}.  We start with the truncated case. Set: 
\[
\alpha_n:=\left\langle \delta_n , \left( \Delta^{(n)}
    +V_{|_{\Z_n}} -\lambda\right)^{-1}\delta_n\right\rangle \in\demi
, 
\] 
where $\Delta^{(n)}$ is the Laplacian on $\Z_n$, see \eqref{e:DeltaNZ}.
As in Proposition \ref{p:rec} we have $\alpha_n=\Phi_n(\alpha_{n+1})$
with 
\[
\Phi_n(
z):=\varphi_{\lambda-V(n),1,1}(z)=-\left(\lambda-V(n)-\left(1+z\right)^{-1}\right)^{-1}
. 
\]
Note that $\Phi_n$ is a contraction of $\demi$. However, unlike in
Proposition \ref{p:cont}, this is not a strict contraction. However,  
$\Phi_n\circ\Phi_{n+1}$ is a strict contraction, see also
\cite{FHS}[Proposition 2.1]. This infers:

\begin{proposition}
Take  $V\in \ell^\infty ({\Z_+},\R)$, then for
all $\lambda\in \demi$ and $(\zeta_n)_n\in\demi^{\Z_+}$ we have 
\[
\dhlim_{n\to\infty}\Phi_0\circ\cdots\circ \Phi_n(\zeta_n)=\alpha_0 .
\]
\end{proposition}
\begin{proof}
See the proof of Corollary \ref{c:cvg} and \cite{FHS}[Theorem 2.3].
\end{proof}
Now unlike in \cite{FHS}[Lemma 4.5] or in \cite{FHS5}[Proposition 3.4]
but as in Proposition \ref{p:born} we use the fixed point of
$\Phi_n\circ\cdots\circ\Phi_{n+\nu-1}$. We obtain: 

\begin{proposition}\label{p:change}
Take $x\in\left( 0,4 \right)$,  $V\in\ell^\infty ({\Z_+},\R )$, and
$\nu\in\Z_+\setminus\{0\}$ with $\lim_{n\to +\infty}V(n)=0$ and $V-\tau^\nu
V\in\ell^1 ({\Z_+},\R )$. Then there exist $x_1,x_2\in\R$ such that 
$x\in (x_1 ,x_2 )$ and $M_1,\e >0$  so that
\[
d_\demi \left( \alpha_0 ,\rmi\right) \leq M_1 
\]
for all $\lambda\in K_{x_1 ,x_2 ,\e}$.
\end{proposition}
Recall that $K_{x_1 ,x_2 ,\e}$ is defined in \eqref{e:K}. 

\begin{proof}
This is the same proof as in Proposition \ref{p:born}. We study the
fixed points of $\Phi_n\circ\cdots\circ\Phi_{n+\nu}$ in a
neighbourhood of 
\[
\omega_{\infty ,x}:=(x,1,1,\ldots ,x,1,1) .
\]
The fixed points of $\varphi_{x,1,1}$ are
\[
-\frac{1}{2}\pm\frac{1}{2}\rmi\sqrt{\frac{4}{x}-1},
\]
The rest remains the same.
\end{proof}
Finally Theorem \ref{t:inac} concludes the proof of Theorem
\ref{t:main0N}. 

We turn to the case of the line. As in Lemma \ref{l:rec}, we reduce
the problem to the case of ${\Z_+}$ because 
\begin{align*}
\left|\left\langle \delta_0 , \left( \Delta^{( \Z )}
    +V -\lambda\right)^{-1}\delta_0\right\rangle\right|&=\left|\left(\lambda
  -V(0)-(1+\alpha_\lambda )^{-1}- (1+\alpha^\prime_\lambda
  )^{-1}\right)^{-1}\right| 
\\
&\leq \frac{1}{\Im(-(1+\alpha_\lambda )^{-1})},
\end{align*}
where
\begin{align*}
& \alpha_\lambda := \left\langle \delta_1 , \left( \Delta^{( \Z_1 )}
    +V_{|_{\Z_1}} -\lambda  \right)^{-1}\delta_1\right\rangle
\in\demi \\ 
& \alpha^\prime_\lambda := \left\langle \delta_{-1} , \left( \Delta^{(
      - \Z_1 )} +V_{|_{-\Z_{-1}}} -\lambda
  \right)^{-1}\delta_{-1}\right\rangle \in\demi , 
\end{align*} 
with $\Delta^{( \Z_1 )}$ and $\Delta^{( -\Z_1 )}$ the Laplacian on
$\Z_1$ and $-\Z_1$ respectively. This gives Theorem \ref{t:main0}.

\end{document}